\documentclass[11pt, onecolumn]{IEEEtran}

\usepackage[margin=1in]{geometry}
\usepackage{mathrsfs}  
\usepackage{amsthm}
\usepackage{amsmath}
\usepackage[colorlinks,citecolor=blue,urlcolor=blue,filecolor=blue,backref=page]{hyperref}
\usepackage{graphicx}
\usepackage{fancyhdr,amssymb,gensymb}
\usepackage[ruled,vlined]{algorithm2e}
\usepackage{caption}
\usepackage{subcaption}
\usepackage{leftidx}
\usepackage{amsfonts,epsfig,latexsym,amscd,multirow,geometry,lscape,bm,pifont}
\usepackage{cancel,algorithmic,verbatim,array,multicol,palatino}
\usepackage{enumitem}
\usepackage[usenames,dvipsnames]{xcolor}
\usepackage{lineno}
\usepackage[normalem]{ulem}
\usepackage{enumitem}
\usepackage{mdframed}
\usepackage{float}
\usepackage[scr=boondox,scrscaled=1.05]{mathalfa}
\usepackage{bbm}
\usepackage{makecell}

\newcommand{\cmmnt}[1]{\ignorespaces}

\newtheorem{theorem}{Theorem}

\newtheorem{remark}{Remark}

\newtheorem{corollary}{Corollary}
\newtheorem{definition}{\noindent \textbf{Definition}}

\newcommand{\C}{\mathcal{C}}

\newcommand{\exE}{\mathbb{E}}

\newcommand{\Bx}{\mathbf{x}}

\newcommand{\BA}{\mathbf{A}}

\newcommand{\BM}{\mathbf{M}}

\newcommand{\BZ}{\mathbf{Z}}

\newcommand{\BIa}{{\boldsymbol{a}}}
\newcommand{\BIb}{{\boldsymbol{b}}}

\newcommand{\BIf}{{\boldsymbol{f}}}
\newcommand{\BIg}{{\boldsymbol{g}}}

\newcommand{\BIk}{{\boldsymbol{k}}}

\newcommand{\BIs}{{\boldsymbol{s}}}

\newcommand{\BIw}{{\boldsymbol{w}}}

\newcommand{\BIy}{{\boldsymbol{y}}}

\newcommand{\CC}{\mathcal{C}}

\newcommand{\CH}{\mathcal{H}}

\newcommand{\CN}{\mathcal{N}}
\newcommand{\CO}{\mathcal{O}}

\newcommand{\CT}{\mathcal{T}}

\newcommand{\CX}{\mathcal{X}}

\newcommand{\BCC}{{\boldsymbol{\mathcal{C}}}}

\newcommand{\IEpsilon}{\epsilon}

\newcommand{\BTheta}{{\boldsymbol{\Theta}}}

\newcommand{\BXi}{{\boldsymbol{\Xi}}}

\newcommand{\BUpsilon}{{\boldsymbol{\Upsilon}}}

\newcommand{\Bgamma}{{\boldsymbol{\gamma}}}

\newcommand{\Btheta}{{\boldsymbol{\theta}}}

\newcommand{\Bmu}{{\boldsymbol{\mu}}}
\newcommand{\Bnu}{{\boldsymbol{\nu}}}

\newcommand{\Bpsi}{{\boldsymbol{\psi}}}

\newcommand{\Bvarphi}{{\boldsymbol{\varphi}}}

\newcommand{\NORMAL}{\CN} 
\newcommand{\GP}{\mathcal{GP}} 
\newcommand{\DIFF}{d} 
\newcommand{\R}{\mathbb{R}} 
\newcommand{\N}{\mathbb{N}} 
\newcommand{\INDI}{\mathbbm{1}} 
\newcommand{\vecone}{\mathbf{1}} 
\newcommand{\veczero}{\mathbf{0}} 
\newcommand{\x}{\Bx} 
\newcommand{\y}{\BIy} 
\newcommand{\tr}{\mathrm{tr}} 
\newcommand{\diag}{\mathrm{diag}} 
\newcommand{\PROB}{\mathbb{P}} 
\newcommand{\EXP}{\mathbb{E}} 
\newcommand{\VAR}{\mathrm{Var}} 
\newcommand{\COV}{\mathrm{Cov}} 

\DeclareMathOperator*{\argmin}{arg\,min} 
\DeclareMathOperator*{\argmax}{arg\,max} 

\begin{document}
\title{Bayesian Spatial Field Reconstruction with Unknown Distortions in Sensor Networks}
\author{
\IEEEauthorblockN{Qikun Xiang$^{1}$, Ido Nevat $^{2}$ and Gareth W. Peters$^3$,
\\
\IEEEauthorblockA{$^1$
School of Physical and Mathematical Sciences, Nanyang Technological University, Singapore}\\
\IEEEauthorblockA{$^2$
TUMCREATE, Singapore \\
\IEEEauthorblockA{$^3$
Department of Actuarial Mathematics and Statistics, Heriot-Watt University, Edinburgh, UK}\\
}
}
}

\maketitle
\begin{abstract}
\noindent
Spatial regression of random fields based on potentially biased sensing information is proposed in this paper. One major concern in such applications is that since it is not known \textit{a-priori} what the accuracy of the collected data from each sensor is, the performance can be negatively affected if the collected information is not fused appropriately. 
For example, the data collector may measure the phenomenon inappropriately, or alternatively, the sensors could be out of calibration, thus introducing random gain and bias to the measurement process.
Such readings would be systematically distorted, leading to incorrect estimation of the spatial field.
To combat this detrimental effect, we develop a robust version of the spatial field model based on a mixture of \textit{Gaussian process} experts. We then develop two different approaches for Bayesian spatial field reconstruction: the first algorithm is the Spatial Best Linear Unbiased Estimator (S-BLUE), in which one considers the quadratic loss function and restricts the estimator to the linear family of transformations; the second algorithm is based on empirical Bayes, which utilises a two-stage estimation procedure to produce accurate predictive inference in the presence of ``misbehaving'' sensors. 
In addition, we develop the distributed version of these two approaches to drastically improve the computational efficiency in large-scale settings.
We present extensive simulation results using both synthetic datasets and semi-synthetic datasets with real temperature measurements and simulated distortions to draw useful conclusions regarding the performance of each of the algorithms.

\end{abstract}

\section{Introduction}

In recent years, Wireless Sensor Networks (WSNs) have attracted considerable attention due to their applications in environment monitoring \cite{nevat2015estimation,sun2018optimal}, forecasting \cite{sheng2018short}, surveillance \cite{sohraby2007wireless}, event detection \cite{soltanmohammadi2013decentralized} and tracking \cite{ristic2013calibration}. 
For example, the United States Environmental Protection Agency (EPA) proposed to promote the use of sensor networks for air quality monitoring \cite{watkins2013draft}. 
In this paper, we focus on environmental monitoring applications in which a WSN consists of a collection of spatially distributed sensor nodes with limited energy and communication bandwidth. The sensors make observations of spatial physical phenomena (e.g.\ the concentration of air pollutants such as carbon monoxide and ozone, temperature, humidity, etc.\ \cite{arfire2015model,sun2018optimal}) and communicate the observations to a Fusion Center (FC) \cite{fazel2012random}. The FC then reconstructs the spatial phenomena from these observations at any spatial location of interest, based on which decisions can be made and actions can be performed. 

In many cases, the sensor nodes used to collect the spatial information are unreliable and distort the spatial information. These distortions should be accounted for in spatial field reconstruction. It is therefore crucial to assess and guarantee the veracity, quality and reliability of collected data \cite{restuccia2017quality, fang2017using}. There are multiple reasons for such a behaviour and here we list three common reasons:
\begin{enumerate}[leftmargin=*]
	\item \underline{Uncalibrated sensors:} uncalibrated sensors, if ignored, can lead to severe degradation of the quality of the fields estimation \cite{karl2015possible}. 
	Traditionally, sensors are calibrated in a controlled environment where the physical input is known and then their performance in a given calibration range is tested and verified over certain operating ranges of the environment, before such WSNs are deployed. This is infeasible for large-scale WSNs due to the prohibitive cost as well as inhomogeneity in deployment schedules. Thus, the calibration has to be done through the so-called \textit{blind} or \textit{self-calibration} techniques  \cite{bilen2014convex}. In addition, the reliability of sensor can deteriorate over time \cite{arfire2015model}, making it very challenging to guarantee the quality of information even for an \textit{a priori} calibrated network. Among the classical calibration models, the \textit{gain-offset} response model is widely-used \cite{saukh2014rendezvous, dorffer2018informed}. 
	We consider the problem of jointly estimating the calibration parameters of individual sensors as well as the spatial field values \cite{eldar2018sensor, cho2018non}.
			\item \underline{Compromised sensors due to malicious intent:} the sensors may be physically compromised such that the sensor observations are maliciously altered to disrupt the operation of the WSN \cite{nurellari2018secure}. 
One common type of attacks is the \textit{Byzantine attack} in which a hostile attacker compromises a part of the sensor network in such a manner that the Fusion Center has imperfect knowledge about whether a sensor node has been compromised~\cite{rawat2010collaborative,zhang2014distributed}. In such a case, there may be erroneous information incorporated into the observations from compromised sensors. 
			\item \underline{Unintentional misuse of sensors:} in many cases data collection is done via crowd-sourcing in which private individuals install sensing stations in order to collect and share their data \cite{monahan2013crowdsourcing}. Since this type of data collection is not performed by professionals, in many cases, the sensors are not placed or used properly, thus introducing distortion into the measurements. Applications of crowd-sensing are becoming common recently due to the ubiquity of the Internet of Things (IoT) \cite{arias2018crowd}. Those sensors could be stationary \cite{sun2018optimal, zhang2018spatial, peters2015utilize, nevat2015estimation, nevat2013random} or mobile \cite{unnikrishnan2013sampling,koukoutsidis2017estimating}, depending on the application.
	\end{enumerate}

Many recent works such as \cite{sun2018optimal,xiang2017trust,zhang2018spatial, peters2015utilize, nevat2015estimation, nevat2013random,unnikrishnan2013sampling} utilise the spatial-temporal correlation to reconstruct spatial physical phenomena at all locations. 
For example, \cite{sun2018optimal} studies the placement of multi-type sensors in Gaussian spatial field to achieve optimal spatial field monitoring. 
To circumvent the threat to data reliability in environment monitoring systems, an estimation procedure was proposed in \cite{xiang2017trust} to detect and exclude malicious sensing agents, while accurately performing spatial field reconstruction. 

The main goal of this paper is to develop statistical procedures that reconstruct spatial fields using observations from sensors with possibly unknown distortions. We refer to such observations as \textit{distorted observations}. We use Gaussian processes as the probabilistic model for spatial phenomena, and the sensors are assumed to follow the \textit{gain-offset} distortion model with multi-modal priors to capture distortion characteristics resulted from different processes, such as natural deterioration, mis-calibration, malicious tampering, and unintentional misplacement or misuse. \\
The main contributions are as follows:
\begin{enumerate}[leftmargin=*]
	\item We develop a two-stage Bayesian inference algorithm that jointly infers the distortions of sensors and reconstructs the spatial field at all locations of interest. The algorithm estimates the distortion parameters in an empirical Bayes manner. 
	\item We derive the posterior distribution and the posterior predictive distribution of the model, and show that the exact computation of Bayes estimators is intractable.
	\item We develop the Spatial Best Linear Unbiased Estimator (S-BLUE) for the model, which is highly computationally efficient.
	\item We solve the optimization problem resulted from empirical Bayes estimation via two efficient methods, the Cross-Entropy method (CEM) and the Iterated Conditional Mode (ICM) method. 
	\item We analyse the computational time complexity of the proposed approaches and develop simple distributed versions of these approaches that are computationally more efficient and suitable for large-scale applications. 
	\item We perform synthetic data experiments as well as an experiment with real-world scenarios to validate our model and estimation procedures. The study with real-world scenarios uses a real temperature dataset from US EPA with synthetically generated distortions to show the real-world applicability of the model. 
\end{enumerate}
The remainder of the paper is organized as follows. We present our Bayesian sensor network model in Section~\ref{sec:SystemModel}, which includes the prior distribution of the distortion parameters. In Section~\ref{sec:bayesian}, we derive the posterior distribution of the parameters as well as the posterior predictive distribution. 
Section~\ref{sec:sblue} introduces the S-BLUE and its properties. 
Section~\ref{sec:ebestimator} introduces the approximation of the Bayes estimators via empirical Bayes, and shows that the maximization of the posterior distribution can be done through CEM and ICM.
Section~\ref{sec:distributed} introduces the distributed approaches.
In Section~\ref{sec:synexp} and Section~\ref{sec:realexp}, we perform experiments using synthetic and real datasets. Finally, Section~\ref{sec:conclusion} concludes the paper. 
\section{Sensor Network Model and Assumptions}
\label{sec:SystemModel}

We begin by presenting the statistical model for the spatial physical phenomena, followed by the system model. 
The following notational convention is used throughout this paper. Boldface upper case symbols denote matrices,  boldface lower case symbols denote column vectors, and standard lower case symbols denote scalars or scalar-valued functions, unless otherwise specified. All vectors are column vectors unless otherwise stated.

\subsection{Spatial Gaussian Random Fields Background}
We model the physical phenomenon as spatially dependent continuous process with a spatial correlation structure. Such models have recently become popular due to their mathematical tractability and accuracy \cite{nevat2015estimation,nevat2013random,nevat2012location,xiang2017trust,sheng2018short,sun2018optimal}.
The degree of the spatial correlation in the process increases with the decrease of the separation between two observing locations and can be accurately modeled as a Gaussian random field\footnote{We use Gaussian Process and Gaussian random field interchangeably.}.
 A Gaussian process (GP) defines a distribution over a space of functions and it is completely specified by the equivalent of sufficient statistics for such a process, and is formally defined as follows.
\begin{definition}
\label{DefGP}
(Gaussian process \cite{rasmussen2005gaussian}):
Let $\mathcal{X} \subset \mathbb{R}^d$ be some bounded
domain of a $d$-dimensional real-valued vector space. Denote by $f(\x): \mathcal{X} \mapsto  \mathbb{R}$ a stochastic process parametrized by $\x \in \mathcal{X}$. Then, the random function $f(\x)$ is a Gaussian process if all its finite dimensional distributions are Gaussian, where for any $m \in \mathbb{N}$, the random vectors $\left(f\left(\x_1\right),\ldots, f\left(\x_m\right)\right)$ has multivariate normal distribution.
\end{definition}
We can therefore interpret a GP as formally defined by the following class of random functions:
\begin{align*}
\begin{split}
&\mathcal{F} :=
\{
f\left(\cdot\right) : \mathcal{X} \mapsto  \mathbb{R}
\;
\text{s.t.}\; f\left(\cdot\right)
\sim\GP\left(\mu \left(\cdot\right),\C\left(\cdot,\cdot\right)\right), \;\\
&\;\mathrm{with}\; \mu\left(\x\right):=\exE\left[f\left(\x\right)\right] : \mathcal{X} \mapsto  \mathbb{R},\\
&\C\left(\x,\x'\right):=
\exE\left[\left(f\left(\x\right)-\mu\left(\x\right)\right)
\left(f\left(\x'\right)-\mu\left(\x'\right)\right)
\right] : \mathcal{X} \times \mathcal{X} \mapsto  \mathbb{R}\},
\end{split}
\end{align*}	
where at each point the mean of the function is $\mu(\cdot)$, and the spatial dependence between any two points is given by the covariance function (Mercer kernel)
$\C \left(\cdot,\cdot\right)$ (see detailed discussion in \cite{rasmussen2005gaussian}).\\

\subsection{Sensor Network System Model}
We begin by presenting the system model followed by the prior distribution specifications. 
\begin{enumerate}
\item[A1.] Consider a random real-valued spatial phenomenon $f:\CX\mapsto\R$ defined on the $d$-dimensional domain $\CX\subset\R^d$. 

\item[A2.] Consider a sensor network with $N$ sensors that sense and transmit data to a Fusion Center (FC) over perfect communication channels. The spatial locations of the sensors, denoted $\left(\x_n\right)_{n=1:N}$ ($\x_n\in\CX, n=1,\ldots,N$), are known at the FC.

\item[A3.] The sensor $n$ transmits $M_n\in\N$ observations to the FC. The observations $(y_{n,m})_{m=1:M_n}$ are generated according to the following \textit{acquisition + distortion} mechanism:
\begin{align}
\begin{split}
\widetilde{y}_{n,m}&=f\left(\x_n\right)+\IEpsilon_{n,m} \;\;\;\;(\text{acquisition})\\
y_{n,m}&=\mathcal{T}\left(\widetilde{y}_{n,m};\Bpsi_n\right)\;\;\;\;\;(\text{distortion})
\end{split}
\end{align}
for $m=1,\ldots,M_n$, where $f\left(\x_n\right)$ is the realisation of the random field at location $\x_n$, $\IEpsilon_{n,m}$ represents the additive random noise at the $n$-th sensor, and $\mathcal{T}:\R\mapsto\R$ is the distortion transformation function, parametrized by $\Bpsi_n$. 

\item [A4.]
The distortion transformation $\mathcal{T}$ has the following generic \textit{gain-offset} form:
\begin{align}
\mathcal{T}\left(u\;;\Bpsi_n = \left(a_n,b_n\right)^T\right):=a_n u+b_n,
\end{align}
where $a_n \in \R_+$ and $b_n \in\R$ 
represent the \textit{gain} and \textit{offset} of the $n$-th sensor, respectively. 
This \textit{gain-offset} model has been widely used to describe sensor characteristics \cite{miluzzo2008calibree,saukh2014rendezvous,dorffer2018informed}.

\item[A5.] 
We assume there are $K+1$ ``categories" of possible distortion transformations. The auxiliary indicator random variable $Z_n$ indicates the category to which each sensor's parameters $\Bpsi_n$ belong. We denote by $Z_n=0$, the \textit{default distortion transformation} category, $\Bpsi_n=\Bpsi^0:=(1,0)^T$ i.e.\ no distortion ($\CT(u,\Bpsi^0)=u$), whereas $Z_n=k, k \in\{1,\ldots,K\}$ indicates that the sensor $n$ belongs to the $k$-th \textit{non-default distortion transformation}.

\end{enumerate}

\subsection{Prior Distribution Specifications}
\label{ssec:modelprior}
\begin{itemize}
	\item[P1.] The spatial random field, $f$, is modelled as a Gaussian process (GP) $F$ with a known mean function $\mu:\CX\mapsto\R$ and a known covariance function $\CC:\CX\times\CX\mapsto\R$, that is,
\begin{align}
F\sim\GP(\mu(\;\cdot\;),\CC(\;\cdot\;,\;\cdot\;)).
\end{align}
 \item [P2.] The additive random noise, $\IEpsilon_{n,m}$, follows a normal distribution with mean zero and a fixed known variance $\varsigma^2$,
\begin{align}
\IEpsilon_{n,m}\overset{\text{i.i.d.}}{\sim}\NORMAL(0,\varsigma^2).
\end{align}
\item[P3.] 
We place a prior distribution on $Z_n$, denoted $\pi({Z_n})$, which is a categorical distribution given by
\begin{align}
\pi(Z_n=k)=q^{(n)}_k\;, k=0,\ldots,K,
\end{align}
where $q^{(n)}_k\ge 0$ for $k=0,\ldots,K$, and $\sum_{k=0}^Kq^{(n)}_k=1$. 

\item [P4.]
Each category of distortion characteristics has a distinct sub-population distribution, denoted $\pi_k$, which translates to the following. For $k=0,\ldots,K$,
\begin{align}
\left(\Bpsi_n|Z_n=k\right)\sim\pi_k,
\end{align}
where $\pi_0$ is a degenerate distribution (or atom) at $\Bpsi^0$,
\begin{align}
\pi_0(\Bpsi_n)=\delta_{\Bpsi^0}.
\end{align}
For $k=1,\ldots,K$, assume that $\pi_k$ has density (and we slightly abuse the notation $\pi_k$ to also denote the density function).
Thus, the prior density of $\Bpsi_n$ (marginalized over $Z_n$), denoted by $\pi(\Bpsi_n)$ is a mixture with an atom at $\Bpsi^0$, given by
\begin{align}
\begin{split}
\pi\left(\Bpsi_n\right)=&q^{(n)}_0\delta_{\Bpsi^0}+\sum_{k=1}^K q^{(n)}_k\pi_k(\Bpsi_n).
\end{split}
\end{align}
\item [P5.]
We assume the independence among $(\Bpsi_n)_{1:N}$, and denote them collectively as $\Bpsi$. Let $\pi({\Bpsi})$ denote the prior density function of $\Bpsi$, which factorizes due to independence,
\begin{align}
\pi(\Bpsi)=\prod_{n=1}^N\pi\left(\Bpsi_n\right).
\end{align}
\end{itemize}

The graphical structure of the proposed Bayesian model is shown in Figure~\ref{fig:dag} as a directed-acyclic-graph (DAG) using plate notations. This graphical illustration is helpful for visualizing the dependencies and conditional independence relations between parameters and random variables. 

\begin{figure}[t]
    \centering
    \includegraphics[width=0.65\linewidth]{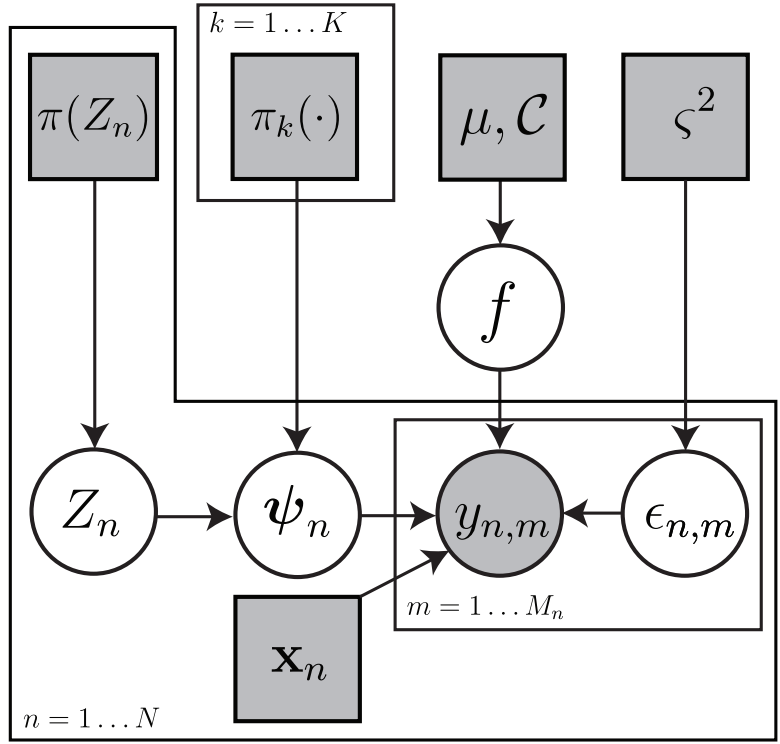}
    \caption{Directed acyclic graph (DAG) of the model using the plate notation. Shaded rectangles represent constants and observed covariates. White circles represent unobserved random variables. The shaded circle represents observed random variable. Arrows represent conditional dependence between two quantities.}
    \label{fig:dag}
\end{figure}

Our objective is to find an estimator $h(\y)$ for $f_*:=f(\x_*)$, the spatial field at location $\x_*$, based on observations $\y:=(y_{n,m})_{n=1:N,m=1:M_n}$. 

\section{Posterior of the Bayesian Model}
\label{sec:bayesian}
In this section we derive the following quantities of interest, based on which the Bayesian estimators will be developed in Sections~\ref{sec:sblue} and \ref{sec:ebestimator}:
\begin{enumerate}[leftmargin=*]
	\item The posterior distribution of the model parameters $\Bpsi$, given by $p(\Bpsi|\y)$ (Theorem~\ref{Th:1}).
\item The posterior predictive distribution $p(f_*|\y)$ (Theorem~\ref{Th:2}).
\end{enumerate}

The following theorem gives the posterior density function of $\Bpsi$.
\begin{theorem}
\label{Th:1}
Let $\Bmu:=\left(\mu\left(\x_n\right)\right)_{1:N}\in\R^N$
be the expected values of the $F$ process at locations $\left(\x_n\right)_{1:N}$, and 
let $\BCC\in\R^{N\times N}$ be the covariance matrix, where $(\BCC)_{ij}=\CC\left(\x_i,\x_j\right)$. 
For $n=1,\ldots,N$, define $g_n:=\sum_{m=1}^{M_n}y_{n,m},\;s_n:=\sum_{m=1}^{M_n}y_{n,m}^2.$
Let $\BM:=\diag[(M_n)_{1:N}]$, $\BIa:=(a_n)_{1:N}$, $\BA:=\diag[\BIa]$, $\BIb:=(b_n)_{1:N}$, $\BIg:=(g_n)_{1:N}$, $\BIs:=(s_n)_{1:N}$. Let $\tilde{g}_n:=a_n^{-1}(M_n^{-1}g_n-b_n)$, $\tilde{\BIg}:=(\tilde{g}_n)_{1:N}=\BA^{-1}(\BM^{-1}\BIg-\BIb)$, $\BUpsilon:=\BCC+\varsigma^2\BM^{-1}$.
Then, the log posterior density function of $\Bpsi$ is given by
\begin{align}
\begin{split}
&\log p(\Bpsi|\y)\\
=&-\frac{1}{2}\left[\tr(\BM)\log2\pi+\tr(\BM\log(\varsigma^2\BA^2))-\log|\varsigma^2\BM^{-1}|\right.\\
&\left.+\log|\BUpsilon|+\varsigma^{-2}\vecone^T\BA^{-2}\BIs-\varsigma^{-2}\BIg^T\BM^{-1}\BA^{-2}\BIg\right.\\
&\left.+(\tilde{\BIg}-\Bmu)^T\BUpsilon^{-1}(\tilde{\BIg}-\Bmu)\right]+\log\pi(\Bpsi)-\log p(\y).
\end{split}
\label{eqn:posterior}
\end{align}
$p(\y)$ is a normalizing constant that is analytically intractable. 
\label{theorem:posterior}
\end{theorem}
\begin{proof}
See Appendix~\ref{apx:proofposterior}.
\end{proof}

\begin{remark}
As shown in the proof of Theorem~\ref{theorem:posterior}, the statistics $g_n$ and $s_n$ are sufficient for~$\Bpsi$. In fact, as we show later, the estimators depend on $\y$ only through $(M_n,g_n,s_n)_{n=1:N}$. Thus, from now on we take $\y:=(M_n,g_n,s_n)_{n=1:N}$ as the summary of observations from all sensors. 
\end{remark}

The following theorem gives the posterior predictive distribution of the model. 
\begin{theorem}
\label{Th:2}
The posterior predictive density is given by
\begin{align}
p(f_*|\y)=\int p(f_*|\y,\Bpsi)p(\Bpsi|\y)\DIFF\Bpsi.
\end{align}
Let $\mu_*:=\mu(\x_*)$, $\CC_*:=\CC(\x_*,\x_*)$, let $\BIk_*:=(\CC\left(\x_n,\x_*\right))_{n=1:N}\in\R^N$ be a column vector, and we have
\begin{align}
p\left(f_*|\y,\Bpsi\right)=&\frac{1}{\sqrt{2\pi \sigma^2_{*}}}\exp\left(-\frac{(f_*-\bar{f}_*)^2}{2\sigma^2_{*}}\right),\label{eqn:condpostpred}\\
\bar{f}_*=&\mu_*+\BIk_*^T\BUpsilon^{-1}\left(\tilde{\BIg}-\Bmu\right),\label{eqn:condpostpredmean}\\
\sigma^2_{*}=&\CC_*-\BIk_*^T\BUpsilon^{-1}\BIk_*.\label{eqn:condpostpredvar}
\end{align}
\label{theorem:postpred}
\end{theorem}
\begin{proof}
See Appendix~\ref{apx:proofpostpred}.
\end{proof}

\section{Spatial Best Linear Unbiased Estimator (S-BLUE)}
\label{sec:sblue}
We now derive the Spatial Best Linear Unbiased Estimator (S-BLUE). 
Let $l(h(\y),f_*)$ denote the loss function, i.e.\ the loss incurred when using estimator $h$ when in fact the target quantity is $f_*$. Let $R[\Pi,h]$ denote the Bayes risk of $h$ associated with the prior distribution $\Pi$, which is defined as the expected value of loss taken over $\Pi$, i.e.
\begin{align*}
R[\Pi,h]=\EXP[l(h(\y),f_*)].
\end{align*}
To derive the S-BLUE, we restrict the estimator to be a member of the family of linear estimators, that is $\CH:=\{h(\y)=\BIw^T\y+b\}$, where $\BIw$ is a weight vector and $b$ is an intercept, both of which do not depend on $\y$ or any unknown variables. Hence, the S-BLUE is defined as the optimal linear estimator under quadratic loss, $l(h(\y),f_*)=\left(h(\y)-f_*\right)^2$, and is given by
\begin{align}
\begin{split}
\widehat{h}_{\text{S-BLUE}}=&\argmin\limits_{h\in\CH}R[\Pi,h]=\argmin\limits_{h\in\CH}\EXP\left[\left(h(\y)-f_*\right)^2\right],
\end{split}
\end{align}
where the expectation is taken over the joint distribution of r.v.'s $(f_*,\y,\Bpsi)$. 
The next theorem shows that $\widehat{h}_{\text{S-BLUE}}$ can be expressed in closed-form. 
\begin{theorem}
$\widehat{h}_{\text{S-BLUE}}$ is given by
\begin{align}
\begin{split}
\widehat{h}_{\text{S-BLUE}}(\y)=&\mu_*+\COV[\bar{\BIg},f_*]^T\COV[\bar{\BIg}]^{-1}(\bar{\BIg}-\EXP[\bar{\BIg}]),
\end{split}
\label{eqn:sblueest}
\end{align}
where $\bar{\BIg}=\BM^{-1}\BIg$. $\EXP[\bar{\BIg}],\COV[\bar{\BIg},f_*],\COV[\bar{\BIg}]$ can all be expressed in closed-form (let $\odot$ denote matrix entry-wise multiplication),
\begin{align}
\EXP[\bar{\BIg}]=&\diag\left(\EXP[\BIa]\right)\Bmu+\EXP[\BIb],\label{eqn:sbluey}\\
\COV[\bar{\BIg},f_*]=&\diag\left(\EXP[\BIa]\right)\BIk_*,\label{eqn:sblueyfstar}\\
\begin{split}
\COV[\bar{\BIg}]=&\EXP[\BIa\BIa^T]\odot\left(\BCC+\varsigma^2\BM^{-1}+\Bmu\Bmu^T\right)\\
&+\diag(\Bmu)\left(\EXP[\BIa\BIb^T]+\EXP[\BIa\BIb^T]^T\right)\\
&+\EXP[\BIb\BIb^T]-\EXP[\bar{\BIg}]\EXP[\bar{\BIg}]^T.
\end{split}\label{eqn:sblueyy}
\end{align}
The various terms in the above equations can all be computed in closed-form, and the details are given in the proof.
\label{theorem:sblue}
\end{theorem}
\begin{proof}
See Appendix~\ref{apx:proofsblue}.
\end{proof}
The next corollary shows the unbiasedness property of $\widehat{h}_{\text{S-BLUE}}$.
\begin{corollary}
$\widehat{h}_{\text{S-BLUE}}$ is unbiased, that is, $\EXP[\widehat{h}_{\text{S-BLUE}}(\y)]=\EXP[f_*]$.
\end{corollary}
\begin{proof}
It is shown via the linearity of expectation.
\end{proof}
The following corollary gives the closed-form expression of $R[\Pi,\widehat{h}_{\text{S-BLUE}}]$ (under quadratic loss).
\begin{corollary}
Under quadratic loss, the Bayes risk associated with $\widehat{h}_{\text{S-BLUE}}$ is given by
\begin{align}
\begin{split}
R[\Pi,\widehat{h}_{\text{S-BLUE}}]=&\CC_*-\COV[\bar{\BIg},f_*]^T\COV[\bar{\BIg}]^{-1}\COV[\bar{\BIg},f_*].
\end{split}
\label{eqn:sbluerisk}
\end{align}
\label{corollary:sblueexploss}
\end{corollary}
\begin{proof}
Substituting (\ref{eqn:sblueest}) into $R[\Pi,h]=\EXP[l(h(\y),f_*)]$ gives the proof. 
\end{proof}

The complete S-BLUE algorithm is shown in Algorithm~\ref{alg:sblue}. Notice that much of the computation in S-BLUE does not require $\y$, and thus can be performed in an ``offline phase'', e.g.\ when the sensor network is deployed. In Algorithm~\ref{alg:sblue}, only Line~\ref{alglin:sblue-inference} needs to be computed when the sensor measurements are taken. Below is a line-by-line analysis of the computational time complexity of Algorithm~\ref{alg:sblue}.

\begin{itemize}[leftmargin=*]
\item Offline phase:
\begin{itemize}
\item Line~\ref{alglin:sblue-ab}: Evaluating $\EXP[\BIa],\EXP[\BIb]$ takes $\CO(KN)$; evaluating $\EXP[\BIa\BIa^T],\EXP[\BIb\BIb^T],\EXP[\BIa\BIb^T]$ takes $\CO(KN^2)$.
\item Line~\ref{alglin:sblue-muc}: Evaluating $\mu_*,\CC_*$ takes $\CO(1)$; evaluating $\Bmu,\BIk_*$ takes $\CO(N)$; evaluating $\BCC$ takes $\CO(N^2)$. 
\item Line~\ref{alglin:sblue-ECov}: Evaluating $\EXP[\bar{\BIg}],\COV[\bar{\BIg},f_*]$ takes $\CO(N)$; evaluating $\COV[\bar{\BIg}]$ takes $\CO(N^2)$. 
\item Line~\ref{alglin:sblue-weight}: Evaluating $\widehat{\BIw}$ takes $\CO(N^3)$, since it involves solving a linear system; evaluating $\widehat{b}$ takes $\CO(N)$. 
\item Line~\ref{alglin:sblue-risk}: Evaluating $R[\Pi,\widehat{h}_{\text{S-BLUE}}]$ takes $\CO(N)$, since $\COV[\bar{\BIg}]^{-1}\COV[\bar{\BIg},f_*]$ has been computed in Line~\ref{alglin:sblue-weight}.
\end{itemize}
\item Online phase:
\begin{itemize}
\item Line~\ref{alglin:sblue-inference}: Evaluating $\widehat{h}_{\text{S-BLUE}}(\y)$ takes $\CO(N)$. 
\end{itemize}
\end{itemize}

Overall, the offline phase of Algorithm~\ref{alg:sblue} takes $\CO(N^3)$ (assuming that $K\le N$), and its online phase takes $\CO(N)$. It is worth noting that there are techniques to further reduce the computational complexity of the matrix inversion, e.g.\ through low-rank approximation (see~\cite{rasmussen2005gaussian}). 

\begin{algorithm}[t]
\KwIn{$\x_*$, $(\x_n)_{n=1:N}$, $\y$, $\left(q^{(n)}_k\right)_{n=1:N,k=0:K}$, $(\pi_k)_{k=1:K}$}
\KwOut{Estimator $\widehat{h}_{\text{S-BLUE}}(\y)$, Bayes risk $R[\Pi,\widehat{h}_{\text{S-BLUE}}]$}
\noindent\rule{6.5cm}{0.4pt} \textit{Offline phase} \noindent\rule{6.5cm}{0.4pt} \\
\nl Compute $\EXP[\BIa],\EXP[\BIb],\EXP[\BIa\BIa^T],\EXP[\BIb\BIb^T],\EXP[\BIa\BIb^T]$ (see Appendix~\ref{apx:proofsblue}). \label{alglin:sblue-ab}\\
\nl Compute $\mu_*,\Bmu,\BIk_*,\BCC,\CC_*$ by their respective definitions. \label{alglin:sblue-muc}\\
\nl Compute $\EXP[\bar{\BIg}],\COV[\bar{\BIg},f_*],\COV[\bar{\BIg}]$ by Equations (\ref{eqn:sbluey}) - (\ref{eqn:sblueyy}). \label{alglin:sblue-ECov}\\
\nl Compute $\widehat{\BIw}=\COV[\bar{\BIg}]^{-1}\COV[\bar{\BIg},f_*]$, $\widehat{b}=\mu_*-\widehat{\BIw}^T\EXP[\bar{\BIg}]$. \label{alglin:sblue-weight}\\
\nl Compute $R[\Pi,\widehat{h}_{\text{S-BLUE}}]$ by (\ref{eqn:sbluerisk}). \label{alglin:sblue-risk}\\
\noindent\rule{6.5cm}{0.4pt} \textit{Online phase} \noindent\rule{6.5cm}{0.4pt} \\
\nl After collecting sensor measurements $\BIy$, compute $\widehat{h}_{\text{S-BLUE}}(\y)=\widehat{\BIw}^T\bar{\BIg}+\widehat{b}$. \label{alglin:sblue-inference}\\
\nl \Return $\widehat{h}_{\text{S-BLUE}}(\y)$, $R[\Pi,\widehat{h}_{\text{S-BLUE}}]$.

    \caption{{\bf Spatial-Best Linear Unbiased Estimator (S-BLUE)}}
    \label{alg:sblue}
\end{algorithm}

\section{Empirical Bayes Estimators}
\label{sec:ebestimator}
We now derive an algorithm in which we do not restrict the estimator to be linear. The idea of empirical Bayes is to plug in a point estimate $\hat{\Bpsi}$ into (\ref{eqn:condpostpred}) to approximate the posterior predictive distribution, i.e.\ $p(\Bpsi|\y)\approx\delta_{\hat{\Bpsi}}$.
This gives us the corresponding empirical Bayes estimators, which minimize the expected posterior loss, conditional on $\hat{\Bpsi}$:
\begin{align}
\label{eqn:bayesestimator}
\widehat{h}_{\text{EB}}(\y,\hat{\Bpsi})=\argmin_{h(\y)}\EXP[l(h(\y),f_*)|\y,\hat{\Bpsi}].
\end{align}
To complete the specification of the estimator we are required to define appropriate \textit{loss functions}.
We present a few widely used loss functions and their corresponding approximate Bayes estimators.

\begin{enumerate}[leftmargin=*]
	\item \textbf{Quadratic loss function:} $l_{\text{quad}}(h(\y),f_*)=(h(\y)-f_*)^2$.
	
The corresponding Bayes estimator is the conditional expectation (minimum mean squared error estimator, or MMSE estimator),
\begin{align*}
\begin{split}
\widehat{h}_{\text{MMSE}}(\y)=&\EXP[f_*|\y]=\int_{\R}f_*p(f_*|\y)\DIFF f_*,
\end{split}
\end{align*}
where $p(f_*|\y)$ is given in Theorem~\ref{theorem:postpred}.
The empirical Bayes version of $\widehat{h}_{\text{MMSE}}(\y)$ is given by
\begin{align*}
\begin{split}
\widehat{h}_{\text{EB-MMSE}}(\y,\hat{\Bpsi})=&\EXP[f_*|\y,\hat{\Bpsi}]=\int_{\R}f_*p(f_*|\y,\hat{\Bpsi})\DIFF f_*\\
\approx&\widehat{h}_{\text{MMSE}}(\y),
\end{split}
\end{align*}
where $p(f_*|\y,\hat{\Bpsi})$ is given in (\ref{eqn:condpostpred}).

\item \textbf{Absolute loss function:} $l_{\text{abs}}(h(\y),f_*)=|h(\y)-f_*|$.

The corresponding Bayes estimator is the conditional median (least absolute deviation estimator, or LAD estimator),
\begin{align*}
\begin{split}
\widehat{h}_{\text{LAD}}(\y)=&\mathrm{median}(f_*|\y).
\end{split}
\end{align*}
The empirical Bayes version of $\widehat{h}_{\text{LAD}}(\y)$ is given by
\begin{align*}
\begin{split}
\widehat{h}_{\text{EB-LAD}}(\y,\hat{\Bpsi})=&\mathrm{median}(f_*|\y,\hat{\Bpsi})\approx \widehat{h}_{\text{LAD}}(\y).
\end{split}
\end{align*}

\item \textbf{$0-1$ loss function:} $l_{\text{0-1}}(h(\y),f_*)=\INDI_{\{f_*<h(\y)\le f_*+\DIFF f_*\}}$.

The corresponding Bayes estimator is the conditional mode (maximum a posteriori estimator, or MAP estimator),
\begin{align*}
\begin{split}
\widehat{h}_{\text{MAP}}(\y)=&\argmax_{f_*}p(f_*|\y).
\end{split}
\end{align*}
The empirical Bayes version of $\widehat{h}_{\text{MAP}}(\y)$ is given by
\begin{align*}
\widehat{h}_{\text{EB-MAP}}(\y,\hat{\Bpsi})=&\argmax_{f_*}p(f_*|\y,\hat{\Bpsi})\approx \widehat{h}_{\text{MAP}}(\y).
\end{align*}
\end{enumerate}
For all the aforementioned Bayes estimators, we first need to find a point estimator for $\Bpsi$. To achieve this, we find the MAP estimator of $\Bpsi$, which aims at maximizing the posterior density $p(\Bpsi|\y)$ given in Theorem~\ref{theorem:posterior}. The MAP estimator is then given by
\begin{align}
\begin{split}
\hat{\Bpsi}=&\argmax_{\Bpsi}p\left(\Bpsi|\y\right)\\
=&\argmax_{\Bpsi}
\left[
-\frac{1}{2}\left\{\tr(\BM)\log2\pi+\tr(\BM\log(\varsigma^2\BA^2))\right.\right.\\
&\left.-\log|\varsigma^2\BM^{-1}|+\log|\BUpsilon|+\varsigma^{-2}\vecone^T\BA^{-2}\BIs\right.\\
&\left.-\varsigma^{-2}\BIg^T\BM^{-1}\BA^{-2}\BIg+(\tilde{\BIg}-\Bmu)^T\BUpsilon^{-1}(\tilde{\BIg}-\Bmu)\right\}\\
&+\log\pi(\Bpsi)\bigg].
\end{split}
\label{eqn:ebobjective}
\end{align}
Note that the optimization objective does not involve $p(\y)$, since it is a constant. Thus, the empirical Bayes estimators could be computed in the following two-stage algorithm:
\begin{enumerate}[leftmargin=*]
\item Compute $\hat{\Bpsi}$ by solving the optimization problem $\argmax_{\Bpsi}p\left(\Bpsi|\y\right)$.
\item Plug in $\hat{\Bpsi}$ to compute $\widehat{h}_{\text{EB}}(\y,\hat{\Bpsi})$.
\end{enumerate}
In order to solve the optimization problem in Step I, we develop two algorithms. The first approach is a stochastic optimization method named Cross-Entropy method (CEM), and the second approach is the Iterated Conditional Modes (ICM) which is based on iterative greedy search.

\subsection{Cross-Entropy Method (CEM)}
\label{ssec:cemethod}

The Cross-Entropy method (CEM) is an stochastic algorithm that is suitable for solving combinatoric or continuous optimization problems. 
Suppose we have a maximization problem with a unique optimizer,
\begin{align*}
\hat{\Bvarphi}=\argmax_{\Bvarphi\in\Phi}J(\Bvarphi),
\end{align*}
where $J(\cdot)$ is the objective function, $\Phi$ is the domain, and $\Bvarphi$ is the parameter vector. 
We solve the optimization problem by considering the level sets of the objective function $\{\Bvarphi:J(\Bvarphi)\ge\gamma\}$, for $\gamma\in\R$. When $\gamma=\hat{J}=\max_{\Bvarphi\in\Phi}J(\Bvarphi)$, we have $\{\Bvarphi:J(\Bvarphi)\ge\gamma\}=\{\hat{\Bvarphi}\}$. 
Next, let us define a family of probability measures $\{\PROB_{\Btheta}:\Btheta\in\BTheta\}$ on $\Phi$ with densities $\{w_{\Btheta}:\Btheta\in\BTheta\}$ that are parameterized by $\Btheta\in\BTheta$. Let $\EXP_{\Btheta}$ denote the expectation taken with respect to $\PROB_{\Btheta}$. Let us fix $\Btheta$ and $\gamma$, and define a rare event probability problem,
\begin{align*}
\PROB_{\Btheta}[J(\Bvarphi)\ge\gamma]=\EXP_{\Btheta}[\INDI_{\{J(\Bvarphi)\ge\gamma\}}]=\int_{\Phi}\INDI_{\{J(\Bvarphi)\ge\gamma\}}w_{\Btheta}(\Bvarphi)\DIFF\Bvarphi.
\end{align*}
Instead of approximating this probability naively by sampling from $w_{\Btheta}$, the importance sampling method is used. Let $w_{\tilde{\Btheta}}$ denote the importance sampler, where $\tilde{\Btheta}\in\BTheta$. Importance sampling approximates the rare event probability by,
\begin{align}
\begin{split}
\PROB_{\Btheta}[J(\Bvarphi)\ge\gamma]=&\int_{\Phi}\INDI_{\{J(\Bvarphi)\ge\gamma\}}w_{\Btheta}(\Bvarphi)\DIFF\Bvarphi\\
=&\EXP_{\tilde{\Btheta}}\left[\INDI_{\{J(\Bvarphi)\ge\gamma\}}\frac{w_{\Btheta}(\Bvarphi)}{w_{\tilde{\Btheta}}(\Bvarphi)}\right]\\
\approx&\frac{1}{S}\sum_{s=1}^S\INDI_{\{J(\tilde{\Bvarphi}^{[s]})\ge\gamma\}}\frac{w_{\Btheta}(\tilde{\Bvarphi}^{[s]})}{w_{\tilde{\Btheta}}(\tilde{\Bvarphi}^{[s]})},
\end{split}
\end{align}
where $\tilde{\Bvarphi}^{[1]},\ldots,\tilde{\Bvarphi}^{[S]}$ are $S$ independent samples generated from $w_{\tilde{\Btheta}}$. The optimal importance sampler $w_{\hat{\Btheta}}$ is selected through the cross-entropy criterion,
\begin{align}
\begin{split}
\hat{\Btheta}=&\argmin_{\tilde{\Btheta}\in\BTheta}\int_{\Phi}\INDI_{\{J(\Bvarphi)\ge\gamma\}}w_{\Btheta}(\Bvarphi)\log\frac{w_{\Btheta}(\Bvarphi)}{w_{\tilde{\Btheta}}(\Bvarphi)}\DIFF\Bvarphi\\
\approx&\argmax_{\tilde{\Btheta}\in\BTheta}\frac{1}{S}\sum_{s=1}^S\INDI_{\{J(\Bvarphi^{[s]})\ge\gamma\}}\log w_{\tilde{\Btheta}}(\Bvarphi^{[s]}),
\end{split}
\label{eqn:cemcecriterion}
\end{align}
where $\Bvarphi^{[1]},\ldots,\Bvarphi^{[S]}$ are $S$ independent samples generated from $w_{\Btheta}$. Notice that the last line of (\ref{eqn:cemcecriterion}) corresponds to the maximum likelihood estimation (MLE) of $\tilde{\Btheta}$ when the samples are $\{\Bvarphi^{[s]}:J(\Bvarphi^{[s]})\ge~\gamma\}$. The CEM starts from an initial sampling distribution $w_{\hat{\Btheta}_0}$ and iteratively updates the threshold $\hat{\gamma}$ and the sampling distribution $w_{\hat{\Btheta}}$. For a detailed introduction of CEM, see \cite{rubinstein1999cross}. The complete procedure is detailed in Algorithm~\ref{alg:cem}.

\begin{algorithm}[t]
\KwIn{number of importance samples $S$, $\rho\in(0,1)$ (typically $0.001\le\rho\le0.01$), initial sampler parameter $\hat{\Btheta}_0$, objective function $J$}
\KwOut{$\hat{\Bvarphi}=\argmax_{\Bvarphi\in\Phi}J(\Bvarphi)$}
\nl $\hat{\gamma}_0\leftarrow-\infty$, $t\leftarrow1$. \label{alglin:cem-init}\\
\nl \Repeat {termination condition is triggered}{
\nl Generate $S$ independent samples $\Bvarphi^{[1]},\ldots,\Bvarphi^{[S]}$ from $w_{\hat{\Btheta}_{t-1}}$. \label{alglin:cem-sample}\\
\nl Compute $J\left(\Bvarphi^{[1]}\right),\ldots,J\left(\Bvarphi^{[S]}\right)$. \label{alglin:cem-score}\\
\nl $\hat{\gamma}_t\leftarrow$ the $(1-\rho)$-sample quantile of $J\left(\Bvarphi^{[1]}\right),\ldots,J\left(\Bvarphi^{[S]}\right)$. \label{alglin:cem-thres}\\
\nl $\hat{\Btheta}_{t}\leftarrow\argmax_{\Btheta\in\BTheta}\frac{1}{S}\sum_{s=1}^S \INDI_{\left\{J\left(\Bvarphi^{[s]}\right)\geq\hat{\gamma}_t\right\}}\log w_{\Btheta}\left(\Bvarphi^{[s]}\right)$.\label{alglin:cem-mle}\\
\nl $t\leftarrow t+1$. \label{alglin:cem-counter}\\
}
\nl Set $\hat{\Bvarphi}$ to be the sample with the largest $J(\hat{\Bvarphi})$ so far. \label{alglin:cem-max}\\
\nl \Return $\hat{\Bvarphi}$.

    \caption{{\bf Cross-Entropy Method (CEM)-Based Optimizer}}
    \label{alg:cem}
\end{algorithm}

We can now link CEM to the MAP estimation problem in (\ref{eqn:ebobjective}). We define the objective function as the (un-normalized) log-posterior conditional density,
\begin{align*}
J(\Bpsi)=&\log p\left(\Bpsi|\y\right)+\log \pi\left(\Bpsi\right).
\end{align*}
For the purpose of demonstrating the CEM, let us assume here that under prior distribution $\pi_k$, $(\log a_n,b_n)$ have a bivariate normal distribution,
\begin{align*}
\left(\begin{smallmatrix}\log a_n\\b_n\end{smallmatrix}\right)\sim\NORMAL(\Bnu_k,\BXi_k).
\end{align*}
Note that this can be easily adapted for other prior distributions. We choose the family of sampling distributions such that,
\begin{align*}
w_{\Btheta}(\Bpsi)=&\prod_{n=1}^Nw_{\Btheta^{(n)}}(\Bpsi_n),\\
w_{\Btheta^{(n)}}(\Bpsi_n)=&r^{(n)}_0\delta_{\Bpsi^0}+\sum_{k=1}^Kr^{(n)}_k\NORMAL((\log a_n,b_n)^T;\tilde{\Bnu}^{(n)}_k,\widetilde{\BXi}^{(n)}_k).
\end{align*}
Here, we have $\Btheta=(\Btheta^{(n)})_{n=1:N}=((r^{(n)}_k)_{0:K}$, $(\tilde{\Bnu}^{(n)}_k)_{1:K}$, $(\widetilde{\BXi}^{(n)}_k)_{1:K})_{n=1:N}$. Before running the CEM algorithm, we set $\hat{\Btheta}_0$ such that $w_{\hat{\Btheta}_0}$ coincides with the prior distribution $\pi$. 
Under this setting, the optimization in Line~\ref{alglin:cem-mle} of Algorithm~\ref{alg:cem} corresponds to the MLE of $\Btheta$, given independent samples $\{\Bpsi^{[s]}:J(\Bpsi^{[s]})\ge\hat{\gamma}_t\}$. This decomposes into sub-problems
\begin{align*}
\hat{\Btheta}^{(n)}=\argmax_{\Btheta^{(n)}}\sum_{s=1}^S\INDI_{\{J(\Bpsi^{[s]})\ge\hat{\gamma}_t\}}\log w_{\Btheta^{(n)}}(\Bpsi^{[s]}_n).
\end{align*}
Since $w_{\Btheta^{(n)}}$ is a mixture distribution, the MLE does not admit a closed-form solution and we use the expectation-maximization (EM) algorithm. 
The EM algorithm is an iterative procedure that computes a local optimum of the likelihood function. 
For notational simplicity, we drop the superscripts and subscripts with $n$ for now, and denote the samples used to obtain the MLE as $\Bpsi^{[1]},\ldots,\Bpsi^{[S]}$. To apply the EM algorithm, let us first introduce the auxiliary variables. Let $z^{[s]}\in\{0,\ldots,K\}$ for $s=1,\ldots,S$ be the discrete auxiliary variables, such that
\begin{align*}
p(\Bpsi^{[s]}|z^{[s]}=0;\Btheta)=&\delta_{\Bpsi^0}\\
p(\Bpsi^{[s]}|z^{[s]}=k;\Btheta)=&\NORMAL((\log a,b)^T;\tilde{\Bnu}_k,\widetilde{\BXi}_k), \\
&\text{ for }k=1,\ldots,K,\\
p(z^{[s]}=k;\Btheta)=&r_k, \text{ for }k=0,\ldots,K.
\end{align*}
This gives the marginal distributions $w_{\Btheta}(\Bpsi^{[s]})$ above. The EM algorithm starts with an initial estimate $\hat{\Btheta}_0$, and iteratively updates the estimated parameter through two steps. In the expectation step (E-step), a lower bound of the log-likelihood function is constructed by first computing the conditional distributions of the auxiliary variables given the estimate of the parameters in the $t$-th iteration
\begin{align*}
p(z^{[s]}=k|\Bpsi^{[s]};\hat{\Btheta}_t)=&\frac{p(\Bpsi^{[s]}|z^{[s]}=k;\hat{\Btheta}_t)p(z^{[s]}=k;\hat{\Btheta}_t)}{\sum_{k'=0}^Kp(\Bpsi^{[s]}|z^{[s]}=k';\hat{\Btheta}_t)p(z^{[s]}=k';\hat{\Btheta}_t)},
\end{align*}
for $k=0,\ldots,K$, 
and then computing the expected value of the log-likelihood function with respect to this conditional distribution, given by
\begin{align*}
Q(\Btheta;\hat{\Btheta}_{t})=\sum_{s=1}^S\EXP_{(z^{[s]}|\Bpsi^{[s]};\hat{\Btheta}_t)}[\log w_{\Btheta}(\Bpsi^{[s]},z^{[s]})].
\end{align*}
In the maximization step (M-step), the estimated parameters in the $(t+1)$-th iteration are computed by maximizing the lower bound $Q(\Btheta;\hat{\Btheta}_{t})$, that is, 
\begin{align*}
\hat{\Btheta}_{t+1}=\argmax_{\Btheta\in\BTheta}Q(\Btheta;\hat{\Btheta}_{t}).
\end{align*}
The algorithm is summarized in Algorithm~\ref{alg:emalg}. For details about the EM algorithm, see \cite{bilmes1998gentle}. 

\begin{algorithm}[t]
\KwIn{$S$ samples $\Bpsi^{[1]},\ldots,\Bpsi^{[S]}$, initial estimate $\hat{\Btheta}_0$}
\KwOut{Estimated parameter $\hat{\Btheta}$}
\nl $\hat{p}\leftarrow-\infty$, $t\leftarrow 0$. \label{alglin:em-init}\\
\nl \Repeat {termination condition is triggered}{ 
\nl \label{alglin:em-loop}\For{$s=1\ldots S$} 
{
\nl Compute $p^{[s]}_k:=p(z^{[s]}=k|\Bpsi^{[s]};\hat{\Btheta}_t)$, for $k=0,\ldots,K$. \label{alglin:em-conditional}\\
}
\nl \textit{(E-step)} Construct $Q(\Btheta;\hat{\Btheta}_{t})=\sum_{s=1}^S\sum_{k=0}^Kp^{[s]}_kp(\Bpsi^{[s]}|z^{[s]}=k;\Btheta)$. \label{alglin:em-E}\\
\nl \textit{(M-step)} $\hat{\Btheta}_{t+1}\leftarrow\argmax_{\Btheta\in\BTheta}Q(\Btheta;\hat{\Btheta}_{t})$. This decomposes into $K+1$ weighted MLE problems. \label{alglin:em-M}\\
\nl $t\leftarrow t+1$. \label{alglin:em-counter}\\
}
\nl $\hat{\Btheta}\leftarrow\hat{\Btheta}_{t}$. \label{alglin:return}\\
\nl \Return $\hat{\Btheta}$.
    \caption{{\bf Expectation-Maximization (EM) Algorithm}}
    \label{alg:emalg}
\end{algorithm}

To analyse the computational time complexity of CEM, let us first assume that each weighted MLE problem in Line~\ref{alglin:em-M} of the EM algorithm (Algorithm~\ref{alg:emalg}) takes $\CO(S)$. For example, this is the case when the sampling distribution is a mixture of normal distributions. The computational time complexity of Algorithm~\ref{alg:emalg} is analysed as follows:
\begin{itemize}[leftmargin=*]
\item Line~\ref{alglin:em-loop}-\ref{alglin:em-conditional}: Each iteration takes $\CO(K)$, the complexity is $\CO(SK)$. 
\item Line~\ref{alglin:em-E}: No actual computation is performed.
\item Line~\ref{alglin:em-M}: Computation of $K+1$ weighted MLE takes $\CO(SK)$.
\end{itemize}
Thus, the total complexity of the EM algorithm is $T_{\text{EM}}=\CO(J_{\text{EM}}SK)$, where $J_{\text{EM}}$ is the number of iterations, which is usually quite small in practice. With this, we analyse the computational time complexity of the CEM estimator as follows:
\begin{itemize}[leftmargin=*]
\item Preparation: 
\begin{itemize}
\item Evaluation of $\log|\BUpsilon|$ and $\BUpsilon^{-1}$ takes $\CO(N^3)$. 
\end{itemize}
\item Algorithm~\ref{alg:cem}:
\begin{itemize}
\item Line~\ref{alglin:cem-sample}: Generation of $S$ samples takes $\CO(SK)$ due to $w_{\hat{\Btheta}_{t-1}}$ being a ($K+1$)-mixture. 
\item Line~\ref{alglin:cem-score}: Evaluation of the objective function (\ref{eqn:ebobjective}) $S$ times takes $\CO(SN^2)$. 
\item Line~\ref{alglin:cem-thres}: Computation of the sample quantile takes $\CO(S)$. 
\item Line~\ref{alglin:cem-mle}: Computation of the maximizer using Algorithm~\ref{alg:emalg} takes $T_{\text{EM}}$. 
\end{itemize}
\item Inference: 
\begin{itemize}
\item Evaluation of $\widehat{h}_{\text{EB}}(\y,\hat{\Bpsi})$ by (\ref{eqn:condpostpredmean}), (\ref{eqn:condpostpredvar}) takes $\CO(N^2)$. 
\end{itemize}
\end{itemize}
The total complexity of the CEM estimator is $\CO(N^3+J_{\text{CEM}}(SN^2+T_{\text{EM}}))$, where $J_{\text{CEM}}$ is the number of iterations in Algorithm~\ref{alg:cem}. Since the EM algorithm converges rather quickly in practice, the complexity of CEM estimator is $\CO(N^3+J_{\text{CEM}}SN^2)$. It is worth noting that the preparation phase of the CEM estimation procedure can be run ``offline'', i.e.\ before having access to sensor measurements. Hence, the CEM estimator has complexity $\CO(N^3)$ in the offline phase, and $\CO(J_{\text{CEM}}SN^2)$ in the online phase. The same remark on the computational complexity of matrix inversion applies here as above.

\subsection{Iterated Conditional Modes (ICM)}
\label{ssec:icm}
We propose a second optimization method to find the MAP estimator $\hat{\Bpsi}$ which is based on iterative greedy search. 
Since $\hat{\Bpsi}\in\R^{2N}$, the dimensionality of the optimization problem is high if $N$ is large. In addition, for $n=1,\ldots,N$, the distribution of $\Bpsi_n$ contains an atom. Hence, to improve the computational efficiency in these settings, we seek to reduce the complexity of the MAP estimation by reducing the global search problem to a sequence of iterative local search problems of iterated conditional modes (ICM)~\cite{besag1986statistical}. 
Let $\Bpsi_{(-n)}:=(\Bpsi_{m})_{m\ne n}$. In each iteration of ICM, we fix $\Bpsi_{(-n)}$ and compute the mode of the conditional posterior distribution $\hat{\Bpsi}_n=\argmax_{\Bpsi_n}p\left(\Bpsi_n|\y,\Bpsi_{(-n)}\right)$ through the conjugate gradient algorithm. ICM converges to a local maximum of the objective function in the sense that $\hat{\Bpsi_n}$ is the mode of the conditional posterior distribution for $n=1,\ldots,N$. 

To optimize the conditional posterior distributions, we decomposed them (up to a normalizing constant) as follows, for $n=1,\ldots,N$,
\begin{align*}
\begin{split}
p\left(\Bpsi_n|\y,\Bpsi_{(-n)}\right)\propto&p\left(\y|\Bpsi\right)\pi\left(\Bpsi_n\right)\\
\propto&p\left(\y_n|\y_{(-n)},\Bpsi\right)\pi\left(\Bpsi_n\right),
\end{split}
\end{align*}
where $\y_n=(y_{n,m})_{m=1:M_n}$, $\y_{(-n)}=(y_{i,m})_{i\ne n,m=1:M_i}$. 
Let $\mu_n:=\mu\left(\x_n\right)$, $\Bmu_{(-n)}:=(\mu_i)_{i\ne n}$. Let $\CC_n:=\CC(\x_n,\x_n)$. Let $\BUpsilon_{(-n,n)}\in\R^{\left(N-1\right)}$ denote the sub-matrix of $\BUpsilon$ involving the cross-terms between sensor $n$ and the rest of sensors. Let $\BUpsilon_{(-n)}\in\R^{\left(N-1\right)\times\left(N-1\right)}$ denote the sub-matrix of $\BUpsilon$ related to sensors other than $n$. Let $\tilde{\BIg}_{(-n)}:=(\tilde{g}_i)_{i\ne n}$.
Completely analogous to Theorem~\ref{theorem:postpred}, we have that,
\begin{align*}
\left(F(\x_n)|\y_{(-n)},\Bpsi\right)\sim\NORMAL\left(\nu_n,\zeta_n\right),
\end{align*}
where
\begin{align}
\begin{split}
\nu_n=&\mu_n+{\BUpsilon_{(-n,n)}}^T{\BUpsilon_{(-n)}}^{-1}\left(\tilde{\BIg}_{(-n)}-\Bmu_{(-n)}\right),
\end{split}\label{eqn:icmcondmean}\\
\begin{split}
\zeta_n=&\CC_n+\varsigma^2-{\BUpsilon_{(-n,n)}}^T{\BUpsilon_{(-n)}}^{-1}{\BUpsilon_{(-n,n)}}.
\end{split}\label{eqn:icmcondcov}
\end{align}
One verifies that $\nu_n$ and $\zeta_n$ do not depend on $\Bpsi_n$.
Therefore, following a derivation similar to that in Theorem~\ref{theorem:posterior}, we have,
\begin{align}
\begin{split}
&\log p\left(\y_n|\y_{(-n)},\Bpsi\right)\\
=&\log\left[\int p(\y_n|f(\x_n),\Bpsi)p(f(\x_n)|\y_{(-n)},\Bpsi)\DIFF f(\x_n)\right]\\
=&-\frac{1}{2}\left[M_n\log2\pi+(M_n-1)\log(\varsigma^2a_n^2)\right.\\
&\left.+\log(a_n^{2}M_n\zeta_n+\varsigma^2a_n^2)+\varsigma^{-2}a_n^{-2}(s_n-M_n^{-1}g_n^2)\right.\\
&\left.+(\zeta_n+\varsigma^2M_n^{-1})^{-1}(\tilde{g}_n-\nu_n)^2\right].
\end{split}
\end{align}
Thus, the log-conditional likelihood as well as its partial derivatives can be efficiently evaluated. The ICM algorithm then separately treats the continuous and discrete parts of the parameter space, that is, comparing $\sup_{\Bpsi_n\neq\Bpsi^0}\log p\left(\Bpsi_n|\y,\Bpsi_{(-n)}\right)$ and $\log p\left(\Bpsi^0|\y,\Bpsi_{(-n)}\right)$. 

The details of the ICM algorithm are shown in Algorithm~\ref{alg:icm}. 
The computational time complexity of the ICM estimator is analysed as follows:
\begin{itemize}[leftmargin=*]
\item Preparation:
\begin{itemize}
\item Evaluation of $\BUpsilon^{-1}$ takes $\CO(N^3)$. 
\end{itemize}
\item Algorithm~\ref{alg:icm}: 
\begin{itemize}
\item Line~\ref{alglin:icm-preloop}-\ref{alglin:icm-prepare}: Evaluation of $\nu_n,\zeta_n$ for $n=1,\ldots,N$ takes $\CO(N^3)$. Notice that the computation of ${\BUpsilon_{(-n)}}^{-1}{\BUpsilon_{(-n,n)}}$ can be simplified to $\CO(N^2)$ via block-wise inversion once $\BUpsilon^{-1}$ has been computed. 
\item Line~\ref{alglin:icm-cg}: Assume that the 2-dimensional optimization takes $\CO(T_{\text{CG}})$. 
\item Line~\ref{alglin:icm-update1}-\ref{alglin:icm-update2}: The complexity is $\CO(1)$. 
\end{itemize}
\item Inference: 
\begin{itemize}
\item Evaluation of $\widehat{h}_{\text{EB}}(\y,\hat{\Bpsi})$ by (\ref{eqn:condpostpredmean}), (\ref{eqn:condpostpredvar}) takes $\CO(N^2)$. 
\end{itemize}
\end{itemize}
The total complexity of the ICM estimator is thus $\CO(N^3+J_{\text{ICM}}NT_{\text{CG}})$, where $J_{\text{ICM}}$ is the number of iterations in Algorithm~\ref{alg:icm}. Similar to CEM, the preparation phase of ICM and Line~\ref{alglin:icm-preloop}-\ref{alglin:icm-prepare} of Algorithm~\ref{alg:icm} can also be run offline. Hence, the ICM estimator has complexity $\CO(N^3)$ in the offline phase, and $\CO(N^2+J_{\text{ICM}}NT_{\text{CG}})$ in the online phase. The same remark on the computational complexity of matrix inversion applies here as above. 

\begin{algorithm}[t]
\KwIn{$(\x_n)_{n=1:N}$, $\y$, $\left(q^{(n)}_k\right)_{n=1:N,k=0:K}$, $(\pi_k)_{k=1:K}$}
\KwOut{Estimation of posterior mode $\hat{\Bpsi}$}
\nl Randomly initialize $\hat{\Bpsi}$.\\
\nl \label{alglin:icm-preloop}\For{$n=1\ldots N$}
{
\nl Compute $\nu_n,\zeta_n$ from (\ref{eqn:icmcondmean}) and (\ref{eqn:icmcondcov}). \label{alglin:icm-prepare}\\
}
\nl \Repeat {termination condition is triggered}{
\nl \For{$n=1\ldots N$}
{
\nl $\tilde{\Bpsi}_n\leftarrow \argmax_{\Bpsi_n\in\R_+\times\R,\Bpsi_n\neq\Bpsi^0} \log p\left(\Bpsi_n|\y,\hat{\Bpsi}_{(-n)}\right)$, by running the conjugate gradient algorithm disregarding the atom at $\Bpsi^0$. \label{alglin:icm-cg}\\
\nl \label{alglin:icm-update1}\If{$\log p\left(\Bpsi^0|\y,\hat{\Bpsi}_{(-n)}\right)<\log p\left(\tilde{\Bpsi}_n|\y,\hat{\Bpsi}_{(-n)}\right)$}{
\nl $\hat{\Bpsi}_n\leftarrow \tilde{\Bpsi}_n$.\\
}\Else{
\nl $\hat{\Bpsi}_n\leftarrow \Bpsi^0$. \label{alglin:icm-update2}\\
}
}
}
\nl \Return $\hat{\Bpsi}$.
    \caption{{\bf Iterative Conditional Modes (ICM) Algorithm}}
    \label{alg:icm}
\end{algorithm}
To account for the multi-modality of the posterior distribution, we adopt a standard multiple start initialization strategy, that is to run ICM from a number of random initial estimates. This corresponds to running Algorithm~\ref{alg:icm} multiple times with different initial values.

\section{Distributed Approaches}
\label{sec:distributed}
Now, let us consider a large scale sensor network with $I\ge 2$ clusters of sensors, where each cluster has a cluster head that locally aggregates data to be sent to the global Fusion Center. The clusters are assumed to be disjoint. We reconstruct the spatial field in a distributed manner. For a spatial location $\x_*\in\CX$, the estimation of $f(\x_*)$ is done in two steps:
\begin{enumerate}[leftmargin=*]
\item Sensors within a cluster transmit the measurements to the cluster head (CH), and the CH performs a local estimation of $f(\x_*)$.
\item The $I$ CHs transmit their local estimations to the FC, where local estimations are fused into the global estimation. 
\end{enumerate}
In this section, we develop fusion algorithms based on local S-BLUE and local empirical Bayes estimators.

\subsection{Distributed S-BLUE}
Suppose that each cluster head $i$ produces the local S-BLUE $\widehat{h}_{\text{S-BLUE}}^{(i)}(\y^{(i)})$ and its Bayes risk $R[\Pi,\widehat{h}_{\text{S-BLUE}}^{(i)}]$, where $\y^{(i)}$ denotes the sensor measurements collected from cluster $i$. The goal is to derive a rule to fuse $\{\widehat{h}_{\text{S-BLUE}}^{(i)}(\y^{(i)})\}_{i=1:I}$ into a single estimator $\widehat{h}_{\text{DS-BLUE}}(\y)$, where the ``D'' in DS-BLUE stands for ``distributed''. We restrict ourselves by considering $\widehat{h}_{\text{DS-BLUE}}(\y)$ as a convex combination of $\{\widehat{h}_{\text{S-BLUE}}^{(i)}(\y^{(i)})\}_{i=1:I}$, that is, $\widehat{h}_{\text{DS-BLUE}}(\y):=\sum_{i=1}^Ic_i\widehat{h}_{\text{S-BLUE}}^{(i)}(\y^{(i)})$, where $c_i\ge0$ for $i=1,\ldots,I$, and $\sum_{i=1}^Ic_i=1$ are the constraints required to preserve the unbiasedness of DS-BLUE. The Bayes risk of $\widehat{h}_{\text{DS-BLUE}}$ is given by,
\begin{align*}
\begin{split}
&R[\Pi,\widehat{h}_{\text{DS-BLUE}}]\\
=&\sum_{i=1}^I\sum_{j=1}^Ic_ic_j\EXP\left[\left(\widehat{h}_{\text{S-BLUE}}^{(i)}(\y^{(i)})-f_*\right)\left(\widehat{h}_{\text{S-BLUE}}^{(j)}(\y^{(j)})-f_*\right)\right]\\
\le&\sum_{i=1}^I\sum_{j=1}^Ic_ic_j\sqrt{R[\Pi,\widehat{h}_{\text{S-BLUE}}^{(i)}]}\sqrt{R[\Pi,\widehat{h}_{\text{S-BLUE}}^{(j)}]}\\
=&\left(\sum_{i=1}^Ic_i\sqrt{R[\Pi,\widehat{h}_{\text{S-BLUE}}^{(i)}]}\right)^2:=\overline{R}[\Pi,\widehat{h}_{\text{DS-BLUE}}],
\end{split}
\end{align*}
where the inequality in the third line above is by the Cauchy-Schwarz inequality. 
The coefficients $(c_i)_{1:I}$ are chosen to minimize the upper bound $\overline{R}[\Pi,\widehat{h}_{\text{DS-BLUE}}]$, which gives the following optimal values,
\begin{align}
c_i^\star=\begin{cases}
1 & \text{if }i=\argmin_{1\le j\le I}R[\Pi,\widehat{h}_{\text{S-BLUE}}^{(j)}],\\
0 & \text{otherwise}.
\end{cases}
\label{eqn:dsbluecoef}
\end{align}
It is assumed above that there is no tie among $(R[\Pi,\widehat{h}_{\text{S-BLUE}}^{(j)}])_{j=1:I}$. If there is one, the tie can be broken arbitrarily. The DS-BLUE is defined via the optimal coefficients in (\ref{eqn:dsbluecoef}), $\widehat{h}_{\text{DS-BLUE}}(\y):=\sum_{i=1}^Ic^\star_i\widehat{h}_{\text{S-BLUE}}^{(i)}(\y^{(i)})$.

The complete DS-BLUE algorithm is shown in Algorithm~\ref{alg:dsblue}. To analyse the computational complexity of Algorithm~\ref{alg:dsblue}, assume for now that each cluster contains at most $N_c$ sensor nodes. The line-by-line analysis of its computational time complexity is as follows:
\begin{itemize}[leftmargin=*]
\item Offline phase:
\begin{itemize}
\item Cluster head: 
\begin{itemize}
\item Line~\ref{alglin:dsblue-dist}: The complexity is $\CO(N_c^3)$, same as the offline phase of Algorithm~\ref{alg:sblue}. 
\end{itemize}
\item Fusion center: 
\begin{itemize}
\item Line~\ref{alglin:dsblue-fus}: Evaluation of $(c_i^\star)_{1:I}$ takes $\CO(I)$. 
\end{itemize}
\end{itemize}
\item Online phase:
\begin{itemize}
\item Cluster head: 
\begin{itemize}
\item Line~\ref{alglin:dsblue-inf-dist}: Evaluating $\widehat{h}_{\text{S-BLUE}}^{(i)}(\y^{(i)})$ takes $\CO(N_c)$. 
\end{itemize}
\item Fusion center:
\begin{itemize}
\item Line~\ref{alglin:dsblue-inf-fus}: Evaluating $\widehat{h}_{\text{DS-BLUE}}(\y)$ takes $\CO(I)$. 
\end{itemize}
\end{itemize}
\end{itemize}
Overall, for cluster heads, the offline phase takes $\CO(N_c^3)$, and the online phase takes $\CO(N_c)$. For the fusion center, both the offline phase and the online phase take $\CO(I)$. 

\begin{algorithm}[t]
\KwIn{$\x_*$, $(\x_n)_{n=1:N}$, $\y$, $\left(q^{(n)}_k\right)_{n=1:N,k=0:K}$, $(\pi_k)_{k=1:K}$}
\KwOut{Estimator $\widehat{h}_{\text{DS-BLUE}}(\y)$, upper bound on the Bayes risk $\overline{R}[\Pi,\widehat{h}_{\text{DS-BLUE}}]$}
\noindent\rule{6.5cm}{0.4pt} \textit{Offline phase} \noindent\rule{6.5cm}{0.4pt} \\
\nl \For{each cluster head $i$ in parallel}{
\nl Compute $\widehat{\BIw}^{(i)}$, $\widehat{b}^{(i)}$, $R[\Pi,\widehat{h}_{\text{S-BLUE}}^{(i)}]$ as in Algorithm~\ref{alg:sblue} and transmit $R[\Pi,\widehat{h}_{\text{S-BLUE}}^{(i)}]$ to the fusion center. \label{alglin:dsblue-dist}\\
}
\nl The fusion center computes the optimal coefficients $(c_i^\star)_{1:I}$ by (\ref{eqn:dsbluecoef}) and the upper bound on the Bayes risk $\overline{R}[\Pi,\widehat{h}_{\text{DS-BLUE}}]$. \label{alglin:dsblue-fus}\\
\noindent\rule{6.5cm}{0.4pt} \textit{Online phase}  \noindent\rule{6.5cm}{0.4pt} \\
\nl \For{each cluster head $i$ in parallel}{
\nl Collect measurements $\y^{(i)}$ from sensors within the cluster $i$. \\
\nl Compute local S-BLUE $\widehat{h}_{\text{S-BLUE}}^{(i)}(\y^{(i)})={{\widehat{\BIw}}^{(i)}}{}^T\bar{\BIg}^{(i)}+\widehat{b}^{(i)}$ and transmit to the fusion center. \label{alglin:dsblue-inf-dist}\\
}
\nl The fusion center computes $\widehat{h}_{\text{DS-BLUE}}(\y)=\sum_{i=1}^Ic^\star_i\widehat{h}_{\text{S-BLUE}}^{(i)}(\y^{(i)})$. \label{alglin:dsblue-inf-fus}\\
\nl \Return $\widehat{h}_{\text{DS-BLUE}}(\y)$, $\overline{R}[\Pi,\widehat{h}_{\text{DS-BLUE}}]$.
    \caption{{\bf Distributed Spatial-Best Linear Unbiased Estimator (DS-BLUE)}}
    \label{alg:dsblue}
\end{algorithm}

\subsection{Distributed Empirical Bayes Estimator}
Similar to DS-BLUE, suppose that each cluster head $i$ computes an approximate posterior distribution $p(f_*|\BIy^{(i)},\hat{\Bpsi}^{(i)})$, where $\y^{(i)}$ denotes the sensor measurements collected from cluster $i$ and $\hat{\Bpsi}^{(i)}$ is a point estimate of the distortion parameters of sensors in cluster $i$. Let $\widehat{h}_{\text{EB-MMSE}}^{(i)}$ denote the empirical Bayes MMSE estimator produced by cluster head $i$. Same as DS-BLUE, let $\widehat{h}_{\text{DEB-MMSE}}(\y):=\sum_{i=1}^Ic_i\widehat{h}_{\text{EB-MMSE}}^{(i)}(\y^{(i)})$ be a convex combination of the local estimators, where $c_i\ge0$ for $i=1,\ldots,I$, and $\sum_{i=1}^Ic_i=1$. We choose the fusion rule to be similar to $\widehat{h}_{\text{DS-BLUE}}$, that is, 
\begin{align}
c_i^\star=\begin{cases}
1 & \text{if }i=\argmin_{1\le j\le I}\VAR[f_*|\BIy^{(j)},\hat{\Bpsi}^{(j)}],\\
0 & \text{otherwise},
\end{cases}
\label{eqn:debcoef}
\end{align}
where ties are broken arbitrarily. 

The complete distributed empirical Bayes algorithm is shown in Algorithm~\ref{alg:dseb}. For the cluster heads, the computational time complexity of this algorithm is the same as in the non-distributed version, with $N$ replaced by $N_c$. For the fusion center, the complexity is $\CO(I)$. 

\begin{algorithm}[t]
\KwIn{$\x_*$, $(\x_n)_{n=1:N}$, $\y$, $\left(q^{(n)}_k\right)_{n=1:N,k=0:K}$, $(\pi_k)_{k=1:K}$}
\KwOut{Estimator $\widehat{h}_{\text{DEB-MMSE}}(\y)$}
\nl \For{each cluster head $i$ in parallel}{
\nl Compute $p(f_*|\BIy^{(i)},\hat{\Bpsi}^{(i)})$ via either CEM or ICM. \\
\nl Compute $\widehat{h}_{\text{EB-MMSE}}^{(i)}(\y^{(i)}),\VAR\left[f_*|\BIy^{(i)},\hat{\Bpsi}^{(i)}\right]$ and transmit to the fusion center. \\
}
\nl The fusion center computes the optimal coefficients $(c_i^\star)_{1:I}$ by (\ref{eqn:debcoef}). \\
\nl The fusion center computes $\widehat{h}_{\text{DEB-MMSE}}(\y)=\sum_{i=1}^Ic^\star_i\widehat{h}_{\text{EB-MMSE}}^{(i)}(\y^{(i)})$. \\
\nl \Return $\widehat{h}_{\text{DEB-MMSE}}(\y)$.
    \caption{{\bf Distributed Empirical Bayes MMSE Estimator (DEB-MMSE)}}
    \label{alg:dseb}
\end{algorithm}

\begin{remark}
\label{rmk:discont}
One downside of $\widehat{h}_{\text{DS-BLUE}}(\y)$ and $\widehat{h}_{\text{DEB-MMSE}}(\y)$ is that the reconstructed spatial field is discontinuous in space. There exists various techniques to avoid discontinuities by smoothing the reconstruction around the discontinuous boundary. However, these are left to future work. 
\end{remark}

\section{Experiments with Synthetic Data}
\label{sec:synexp}
We conduct two experiments with synthetically generated data to study the performance of the methods we proposed including S-BLUE, CEM, and ICM. In Section~\ref{ssec:sensitivity}, we study the sensitivity of the proposed methods to the strength of distortions. In Section~\ref{ssec:realistic}, we perform a realistic simulation and analyse the overall performance of the proposed methods. 
In the studies, the two empirical Bayes-based estimators (CEM and ICM) use the quadratic loss function and hence correspond to $\widehat{h}_{\text{EB-MMSE}}$. 
The proposed methods are compared to two baselines, the ``oracle'' case in which the distortions $\Bpsi$ are known exactly, and the ''naive'' case in which distortions are disregarded in the prediction, i.e.\ $\hat{\Bpsi}_n=\Bpsi^0$, for $n=1,\ldots,N$. 

\subsection{Synthetic Experiment 1: Homogeneous Distortion Characteristics}
\label{ssec:sensitivity}
In this experiment, we study an ideal scenario where the signal-to-noise ratio (SNR) is high and observations are plentiful. We fix the number of observations per sensor to be 50, and the SNR to be 15dB. Notice that for the ease of comparison, all SNRs are measured at the sensor level, that is, the SNRs of aggregated observations. Under the i.i.d. noise assumption, we define $\text{SNR}=10\log_{10}\left(\frac{M\VAR(F)}{\varsigma^2}\right)$, where $\VAR(F)$ denotes the signal variance, $\varsigma^2$ denotes observation noise variance, and $M$ denotes the number of observations per sensor. 

We simulate a spatial field defined on the two-dimensional square $\CX=[0,1]^2$ with mean 10, i.e.\ $\forall \x\in\CX$, $\mu(\x)=10$, 
and a Mat\'ern covariance function with $\nu=3/2$, $\VAR(F)=100$, and length scale=0.3, i.e.\ $\forall \x,\x'\in\CX$,
$\CC(\x,\x')=100\left(1+\frac{\sqrt{3}\|\x-\x'\|}{0.3}\right)\exp\left(-\frac{\sqrt{3}\|\x-\x'\|}{0.3}\right).$ Here $\|\x-\x'\|$ corresponds to the Euclidean distance. The contour plot of the simulated spatial field is shown in Figure~\ref{sfig:sim1sp}. 

Subsequently, 100 sensors are randomly placed in the square. 50 out of the 100 sensors are fixed to have identical distortion parameters, and the rest are set to have the default transformation parameters $\Bpsi^0$, that is, non-distorting. For the sensors with distortions, we first fix the offset parameter $b_n$ at 5, and vary the gain parameter $a_n$ from 1 to 1.6. Then we fix the gain parameter at 1.2, and vary the offset parameter from 0 to 12. With each setting of distortion parameters, 100 sets of noisy observations are randomly simulated. For each set of observations, the three proposed methods: S-BLUE, CEM, and ICM, along with the two baselines oracle and naive, are used to reconstruct the spatial field. Here, weakly informative prior for the distortion parameters is used, which has a single category ($K=1$) given by $q^{(n)}_1=0.5$ for all $n$, where under $\pi_1$, $a_n\sim\log\NORMAL(0.25, 0.1^2)$, $b_n\sim\NORMAL(6, 3^2)$. 
The reconstruction accuracy is evaluated by the mean-squared-errors (MSE) at a $100\times 100$ grid on $[0,1]^2$. 

Figure~\ref{fig:sim1} shows the reconstruction accuracy averaged over 100 realizations. For better interpretability, the ratio between the MSE and the prior variance, referred to as the relative MSE, is shown. Error bars in Figure~\ref{fig:sim1} indicate the 95\% Student's $t$-confidence interval of the relative MSE estimated from the 100 realizations. Error bars in all subsequent figures indicate the 95\% Student's $t$-confidence interval of the respective underlying quantity. From Figure~\ref{fig:sim1}, one observes that the MSE of the oracle stayed constant, as expected, while the MSE of the naive baseline increased rapidly when the distortion parameters increased. 
The MSE of S-BLUE first decreased and then increased slightly. This is due to the way the prior distribution of distortions were set up. The prior mean of $a_n$ and $b_n$ were 1.29 and 6.0, respectively. Since S-BLUE makes predictions based purely on the prior information, its performance is best when the actual distortion is closest to the prior mean. 
The two empirical Bayes-based methods showed decreasing MSE when the gain parameter increased and slightly increasing MSE when the offset parameter increased. The reason is that since the gain parameter affects both the location and the spread of the observations, while the offset parameter only affects the location, the gain in the distortion was more noticeable and thus easier to detect. The MSEs of CEM and ICM were almost identical. Figure~\ref{fig:sim1r} shows the average false positive rate (FPR) and false negative rate (FNR) of CEM and ICM. The FPR is defined as the proportion of non-distorting sensors that were estimated to be distorting, and the FNR is defined as the proportion of distorting sensors that were estimated to be non-distorting. The FPR and FNR of the two methods were almost identical. Notice in addition that the FNR was high when the gain was 1 and the offset was 5. This was caused by the short length scale (0.3), which made detection of the offset hard due to the low spatial correlation. Finally, all error bars are narrow, indicating that difference between the performance of different methods are statistically significant. To further confirm this, we show the maximum absolute deviation from the relative MSE of the five methods in the second column of Table~\ref{tab:simmaxdev}. One checks that the deviations are small, indicating that the performance is stable across realizations. 

\begin{figure}[t]
\centering

\begin{subfigure}[b]{0.48\linewidth}
\includegraphics[width=\linewidth]{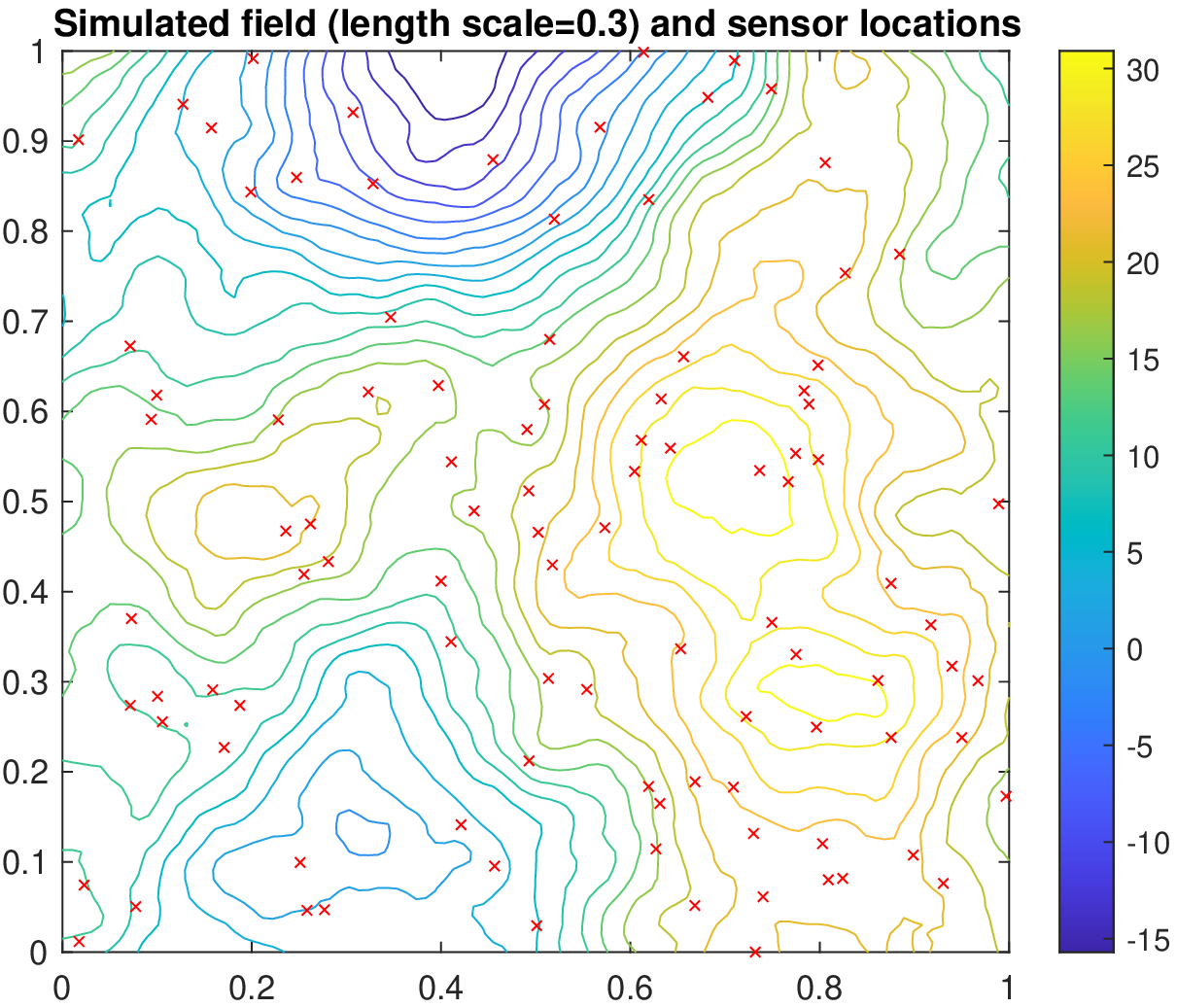}
\caption{Experiment 1}
\label{sfig:sim1sp}
\end{subfigure}
~
\begin{subfigure}[b]{0.48\linewidth}
\includegraphics[width=\linewidth]{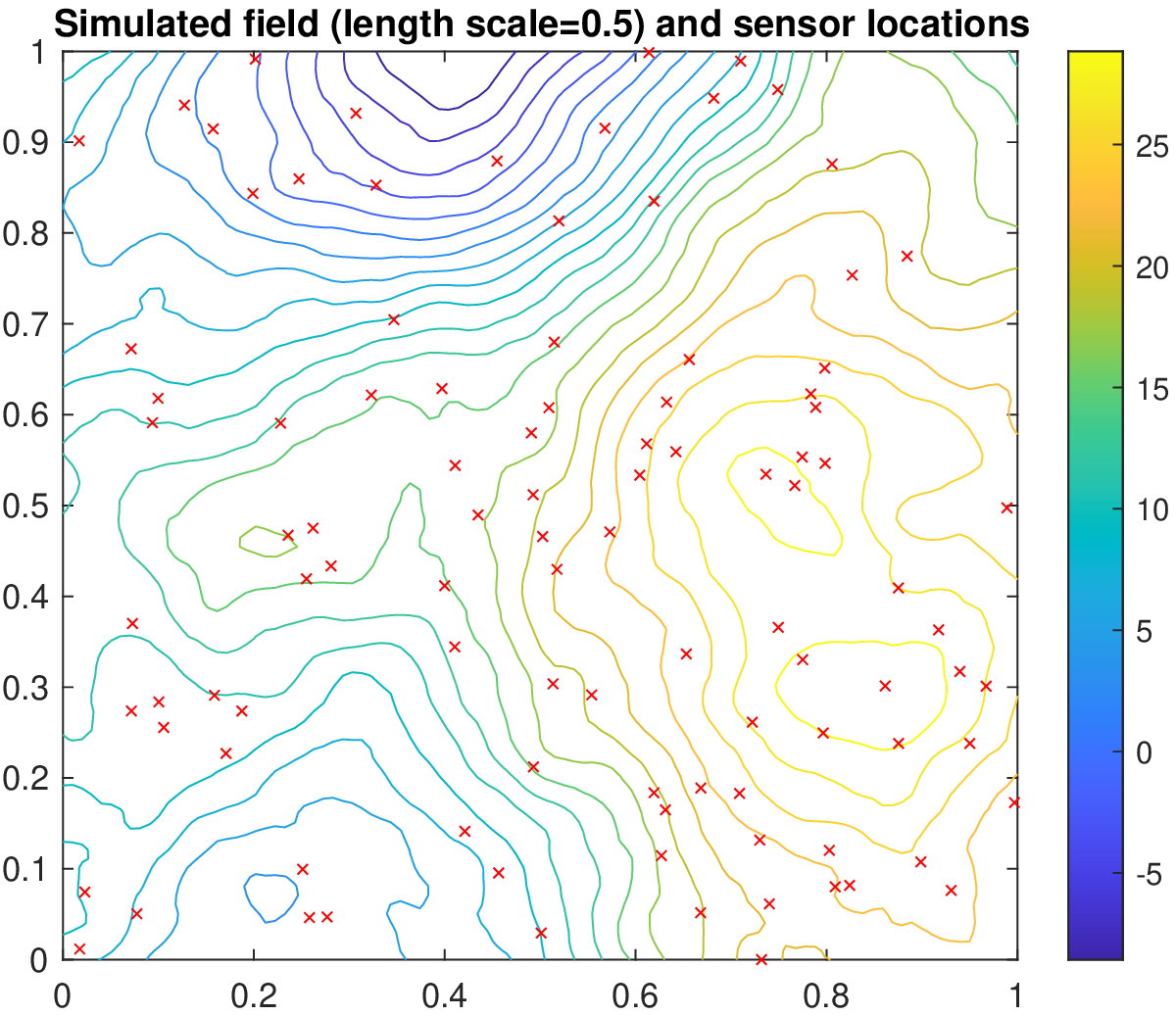}
\caption{Experiment 2}
\label{sfig:sim2sp}
\end{subfigure}
\caption{Contour plots of simulated spatial fields used in the two synthetic experiments with sensor locations.}
\label{fig:simsp}
\end{figure} 

\begin{figure}[t]
\centering
\includegraphics[width=\linewidth]{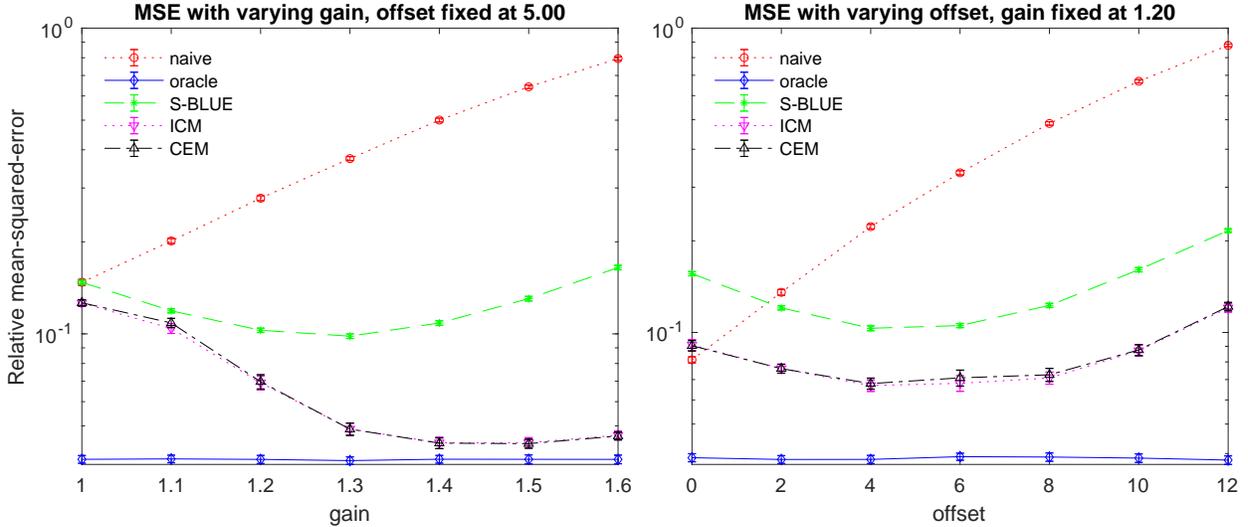}
\caption{Synthetic experiment 1 -- relative MSE (log-scale) with error bars indicating the 95\% confidence interval against varying strengths of distortion.}
\label{fig:sim1}
\end{figure}

\begin{figure}[t]
\centering
\includegraphics[width=\linewidth]{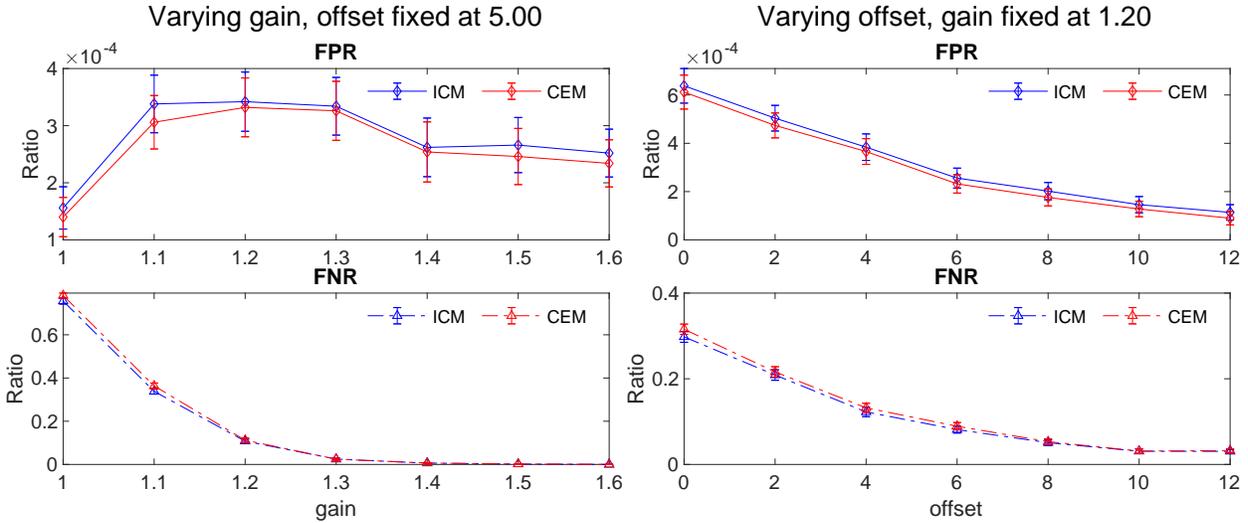}
\caption{Synthetic experiment 1 -- FPR and FNR of CEM and ICM with error bars indicating the 95\% confidence interval against varying strengths of distortion.}
\label{fig:sim1r}
\end{figure}

\subsection{Synthetic Experiment 2: Inhomogeneous Distortion Characteristics}
\label{ssec:realistic}
In the second synthetic experiment, we study a realistic scenario where each sensor has different distortion parameters, and vary the SNR as well as the number of observations. 

We again simulate a spatial field defined on the two-dimensional square $[0,1]^2$. This time, however, the length scale is set to be 0.5, and the spatial correlation decays at a slower rate. The contour plot of the simulated spatial field is shown in Figure~\ref{sfig:sim2sp}. The 100 sensors are placed at the same locations as in the synthetic experiment~1. 50 out of the 100 sensors are randomly selected to have the different distortion parameters generated from the following prior with three categories ($K=3$), given by $q^{(n)}_1=q^{(n)}_2=q^{(n)}_3=\frac{1}{6}$ for all $n$, where under $\pi_1$, $a_n\sim\log\NORMAL(-0.4, 0.05^2)$, $b_n\sim\NORMAL(0, 0.2^2)$, under $\pi_2$, $a_n\sim\log\NORMAL(0.2, 0.05^2)$, $b_n\sim\NORMAL(0, 0.2^2)$, under $\pi_3$, $a_n\sim\log\NORMAL(0, 0.05^2)$, $b_n\sim\NORMAL(10, 2^2)$,
and the rest of the sensors are set to be non-distorting. 
The distortion parameters are generated and fixed in this experiment. After that, we randomly simulate 100 sets of noisy observations. 
For each set of observations, we test the proposed methods along with the baselines as in the synthetic experiment~1. 

\begin{table}[t]
\centering
\caption{Synthetic experiment 1 \& 2 -- maximum absolute deviation from the average relative MSE of the five methods. }
\label{tab:simmaxdev}
\begin{tabular}{|l|r|r|}
\hline 
Method & Experiment 1 & Experiment 2 \\ 
\hline 
oracle & 0.0278 & 0.1615 \\ 
\hline 
naive & 0.1566 & 0.2064 \\ 
\hline 
S-BLUE & 0.0410 & 0.1800 \\ 
\hline 
ICM & 0.0937 & 0.1935 \\ 
\hline 
CEM & 0.0908 & 0.1928 \\ 
\hline 
\end{tabular} 
\end{table}

Figure~\ref{fig:sim2} shows the relative MSE averaged over 100 realizations, with different number of observations per sensor and different SNR. Figure~\ref{fig:sim2r} shows the FPR and FNR of CEM and ICM averaged over 100 realizations. Observe that the MSE of the baselines and S-BLUE did not change with different number of observations per sensor because they depend only on the mean of observations from each sensor. CEM and ICM, on the other hand, benefited from having access to more observations. The naive baseline showed a peculiar trend that first decreased and then increased when SNR increased. The reason is that the naive estimator is a linear combination of the prior mean and the observations, where the weights depend on the SNR. With high SNR, the naive estimator placed a high weight on the distorted observations, thus making the MSE high. The MSE of S-BLUE decreased steadily when SNR increased, and eventually flattened out. The MSE of CEM and ICM depended highly on the number of observations. CEM had the best reconstruction quality among the proposed methods when observations were plentiful or when the SNR was high. In comparison, ICM performed considerably worse than CEM when the number of observations was 5 and 20. This was due to the higher spatial correlation, which made the dependency between distortion parameters higher in the posterior and reduced the effectiveness of iterative greedy search. Notice that the FPR of CEM and ICM was close to 1 when the SNR was low and the observations were scarce. This indicates that CEM and ICM estimated all of the sensors as distorting. The reason was that with high noise variance, the likelihood had a flat shape, and thus the posterior mode did not contain non-distorting sensors. It can also be observed that the FPR and FNR were low when the number of observations was 100 and the SNR was high. 
Same as in the synthetic experiment~1, the error bars are narrow, indicating the statistical significance of the differences. The third column of Table~\ref{tab:simmaxdev} shows the maximum absolute deviation from the relative MSE in the synthetic experiment~2. The deviations are larger compared to the synthetic experiment~1 due to the case with small number of samples and low SNR. Nonetheless, this shows that the performance of the methods is stable across realizations. 
To further demonstrate the performance of the proposed methods under different distortion parameters, we repeat the above experiment. This time, instead of fixing the distortion parameters across realizations, the distortion parameters are independently randomly generated in each of the 100 realizations. The result is shown in Figure~\ref{fig:sim2rand}. Since Figure~\ref{fig:sim2} and Figure~\ref{fig:sim2rand} look very similar, we confirm that the performance is stable across a range of distortion parameters. 

Finally, to show the effect of the proportion of distorting sensors, an additional experiment is performed. Figure~\ref{fig:sim2ar} shows the relative MSE of S-BLUE, CEM, ICM and the baselines when the number of observations is fixed at 100, the SNR is fixed at 20dB, and the proportion of distorting sensors varies from 0 to 1. These results were averaged across 100 independent realizations of distortion parameters and random noises. From Figure~\ref{fig:sim2ar}, it can be observed that when no sensor was distorting, the relative MSE of all methods coincided. However, as the proportion of distorting sensors increased, the relative MSE of the naive baseline increased rapidly, the relative MSE of S-BLUE increased slowly, the relative MSE of CEM and ICM increased only slightly and flattened out eventually, and the oracle, as expected, was unaffected by the change of proportion. The reason is that increasing the proportion of distorting sensors resulted in an increase in prior model uncertainty that negatively affected S-BLUE, while CEM and ICM used the information from the observations to estimate the distortion parameters and thus were only slightly affected. 

\begin{figure}[t]
\centering
\includegraphics[width=\linewidth]{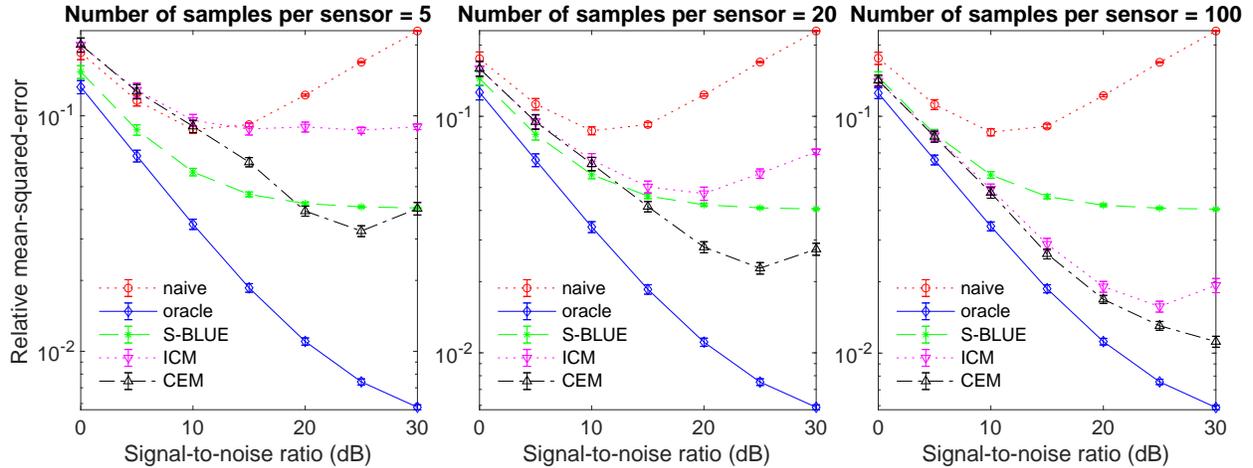}
\caption{Synthetic experiment 2 -- relative MSE against varying number of observations and SNR.}
\label{fig:sim2}
\end{figure}

\begin{figure}[t]
\centering
\includegraphics[width=\linewidth]{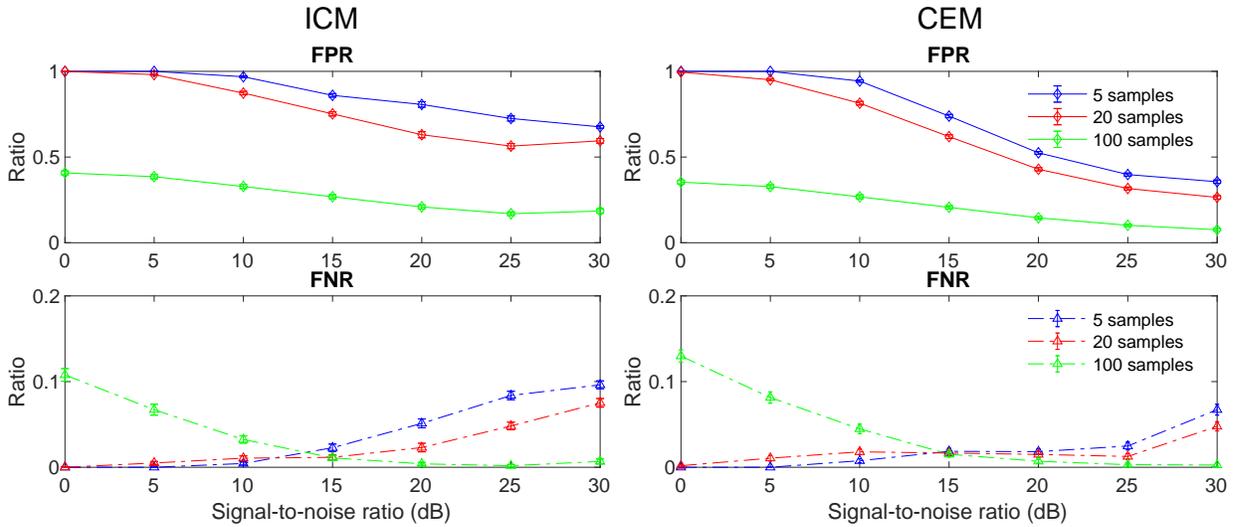}
\caption{Synthetic experiment 2 -- FPR, FNR of CEM and ICM against varying number of observations and SNR.}
\label{fig:sim2r}
\end{figure}

\begin{figure}[t]
\centering
\includegraphics[width=\linewidth]{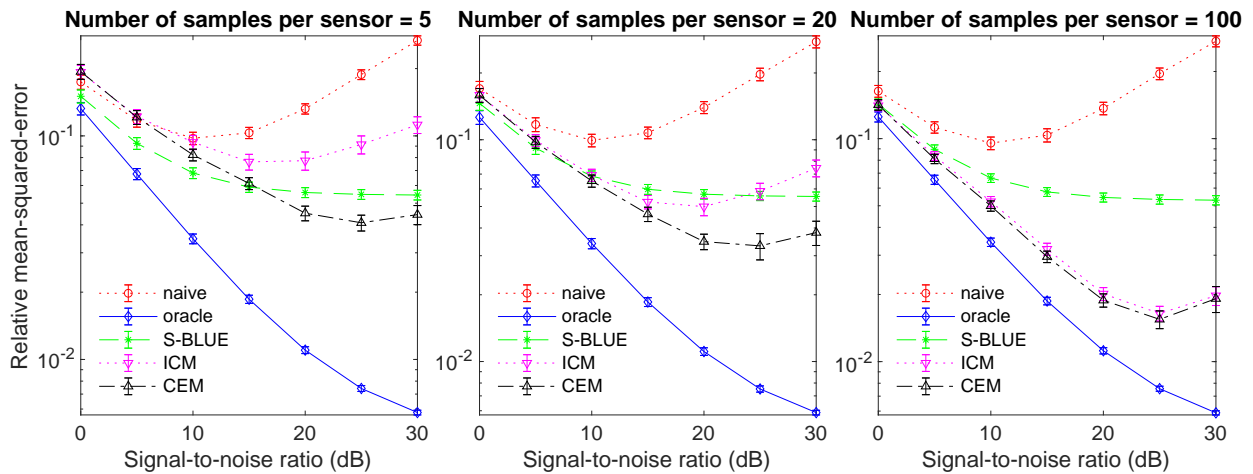}
\caption{Synthetic experiment 2 -- relative MSE against varying number of observations. Distortion parameters were randomly generated in each of the 100 realizations. }
\label{fig:sim2rand}
\end{figure}

\begin{figure}[t]
\centering
\includegraphics[width=0.5\linewidth]{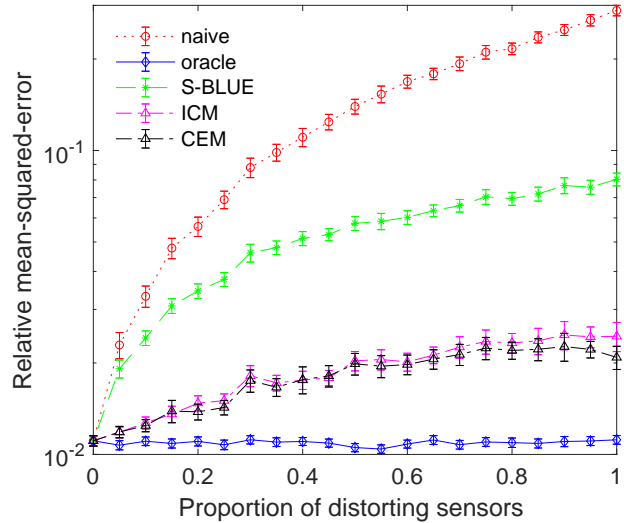}
\caption{Synthetic experiment 2 -- relative MSE against varying proportion of distorting sensors. The number of observations is fixed at 100 and the SNR is fixed at 20dB. }
\label{fig:sim2ar}
\end{figure}

\section{Experiments with Real Data}
\label{sec:realexp}
We study the 2017 US temperature dataset from the US EPA\footnote{\url{https://aqs.epa.gov/aqsweb/airdata/download_files.html}, retrieved on 5 June 2018.}. The dataset contains 309,226 rows, with the following fields:
\begin{itemize}[leftmargin=*]
\item \texttt{State.Code}: the numerical code of the state.
\item \texttt{County.Code}: the numerical code of the county.
\item \texttt{Site.Num}: the numerical code of the monitoring site.
\item \texttt{Longitude}: the longitude of the site.
\item \texttt{Latitude}: the latitude of the site. 
\item \texttt{Date.Local}: the local date on which the temperature measurement was taken.
\item \texttt{X1st.Max.Value}: the maximum hourly temperature measurement of a day.
\end{itemize}

\subsection{Preprocessing}
The first step of preprocessing is to remove the irrelevant fields from the dataset. For temperature measurements, we take the maximum hourly measurement of every day. The measurements are converted from Fahrenheit into Celsius. We noticed that there is an obvious outlier which corresponds to 125$\degree$C (the next highest measurement is 55$\degree$C) with state code 6, county code 79, site number 5 on day~161. 
We also noticed another potential outlier corresponding to $-17.89\degree$C ($0\degree$F) in June, while the next lowest temperature from May to July is 1.1$\degree$C. This measurement has state code 38, county code 93, site number 101 on day 181. Therefore, these two measurements are removed from the dataset. 

In the dataset, there are 830 monitoring sites in total. Figure~\ref{fig:usmap} is a scatter plot showing the spatial locations of these monitoring sites. 
Not all monitoring sites have taken measurements on all 365 days. We refer to these as missing measurements. Out of the 830 monitoring sites, 55 contain more than 30\% of missing measurements. Overall, 16,980 measurements are missing, which is 5.6\% of the 302,950 measurement (365$\times$830). 
\begin{remark}
We have noticed that some monitoring sites have reported temperature measurements by different types of instrument, which explains why we have slightly fewer measurements after we sort measurements into the site \& date format. However, measurements on the same day at the same monitoring site with different instruments tend to be close hence we choose an arbitrary set whenever a monitoring site reports multiple sets of measurements. 
\end{remark}
\begin{figure}[t]
\centering
\includegraphics[width=0.7\linewidth]{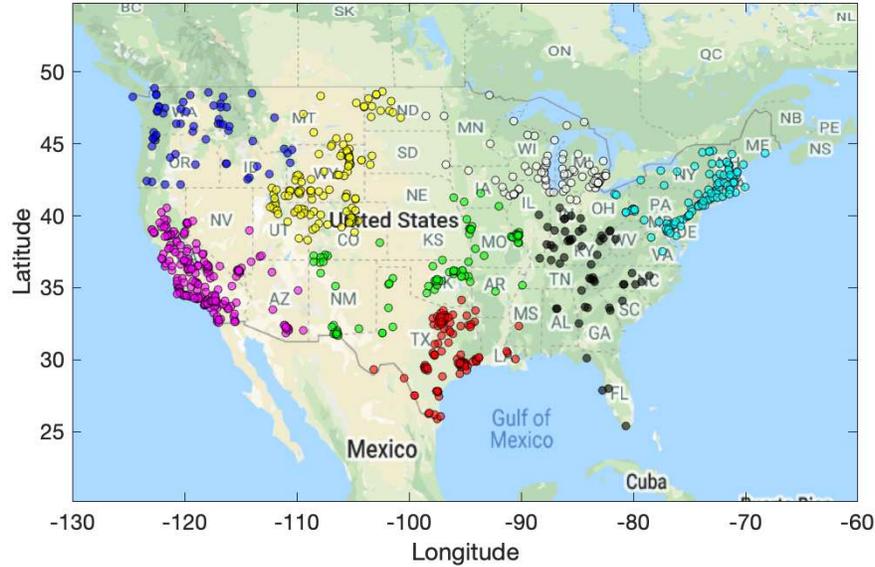}
\caption{Spatial locations of temperature monitoring sites. Colors of the points indicate the 8 spatial clusters used to evaluate the distributed version of the proposed methods. }
\label{fig:usmap}
\end{figure}

\begin{figure}[t]
\centering

\begin{subfigure}[b]{0.48\linewidth}
\includegraphics[width=\linewidth]{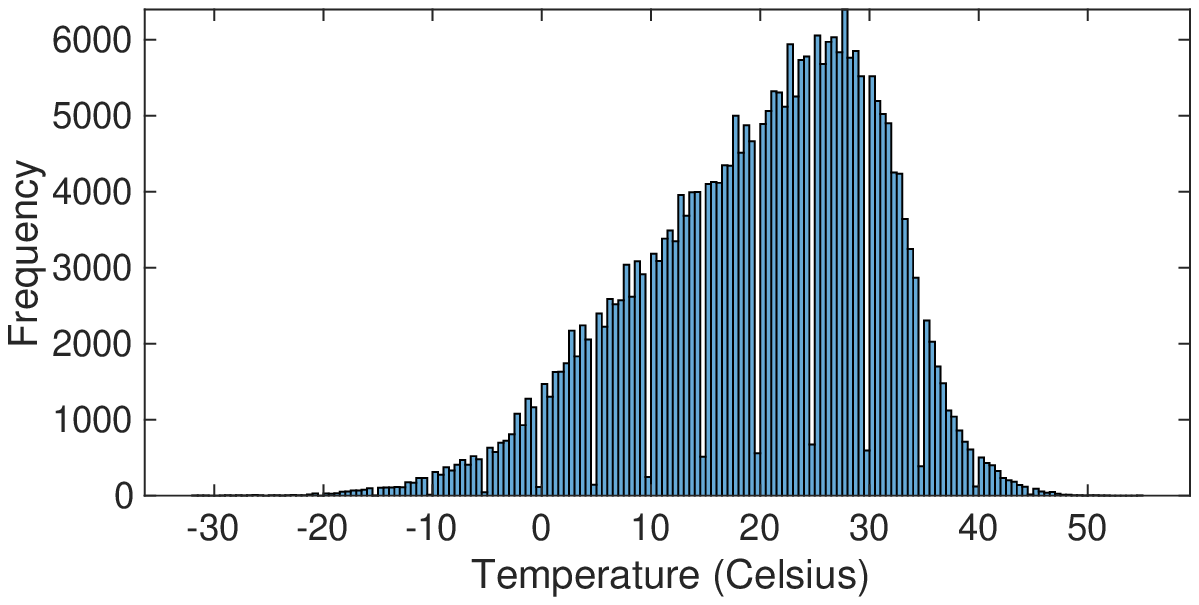}
\caption{histogram of raw data}
\label{sfig:ustemphist1}
\end{subfigure}
~
\begin{subfigure}[b]{0.48\linewidth}
\includegraphics[width=\linewidth]{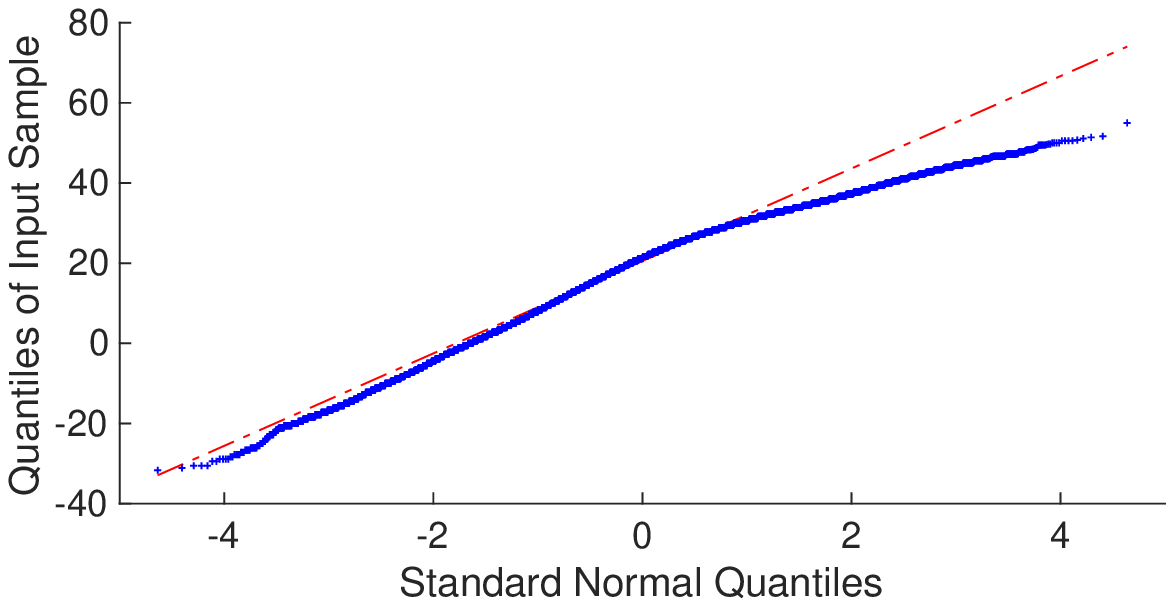}
\caption{normal Q-Q plot}
\label{sfig:ustempqq}
\end{subfigure}
\caption{Histogram and normal Q-Q plot of the temperature measurements.}
\label{fig:ustempdesc}
\end{figure}

Figure~\ref{fig:ustempdesc} shows the histogram and the normal quantile-quantile (Q-Q) plot of the measurements. From Figure~\ref{sfig:ustemphist1}, one can see that the measurements contain quantization artifacts. Nonetheless, the measurements are treated as continuous. From Figure~\ref{sfig:ustempqq}, one sees that the dataset does not contain abnormally large or small values, and is slightly left-skewed compared to a normal distribution. Overall, the normality assumption holds approximately, as we have not yet considered the seasonal shift of temperature. 
Table~\ref{tab:ustempstats} shows the summary statistics of these measurements. 
\begin{table}[t]
\caption{Summary statistics of the temperature measurements.}
\begin{center}
\begin{tabular}{|r|r|r|r|r|}
\hline
mean & variance & skewness & kurtosis & \\
\hline
19.7930 & 119.9365 & -0.5429 & 2.9377 &\\
\hline
min & 1st-quartile & median & 3rd-quartile & max\\
\hline
-31.6667 & 12.7778 & 21.1111 & 28.3333 & 55.0000\\
\hline
\end{tabular}
\end{center}
\label{tab:ustempstats}
\end{table}

\subsection{Smoothing of Spatial Field}
The first step of the spatial analysis is to model the daily observations as noisy samples from a Gaussian process without distortions and estimate its hyperparameters. For subsequent analyses, we take only data from the 20 days between day 181 to day 200 from the dataset, and restrict ourselves to the contiguous United States (that is, excluding Hawaii and Alaska) since Hawaii and Alaska are far away from the rest of the Unites States. As for the covariance function, we again choose to use the Mat\'ern covariance function with $\nu=3/2$ to allow for high flexibility while keeping the spatial field mean-square differentiable. For each calendar day, we estimate the signal mean, signal variance, noise variance and length-scale via maximum marginal likelihood estimation (see Chapter~5 of \cite{rasmussen2005gaussian} for details) using all available observations on that day. Then, the median of the estimated values on 20 days are taken as the estimated hyperparameters of the GP. The estimated hyperparameters are shown in Table~\ref{tab:ustemphyp}.

\begin{table}[t]
\caption{Estimated GP hyperparameters}
\begin{center}
\begin{tabular}{|r|r|}
\hline
signal mean (\degree C) & signal std. dev. (\degree C) \\
\hline
30.5034 & 4.6587 \\
\hline
noise std. dev. (\degree C) & length-scale (km) \\
\hline
1.6340 & 174.3699 \\
\hline
\end{tabular}
\end{center}
\label{tab:ustemphyp}
\end{table}

Using the estimated hyperparameters, we reconstruct the spatial field for each calendar day at both a $100\times100$ grid of locations and the 824 sensor locations by computing the posterior mean and the posterior covariance matrix, conditional on all the available observations on that day. Since the posterior mean is usually much smoother compared to the actual spatial field, a sample spatial field is generated from the posterior distribution to make it realistic. The generated spatial field at a grid of locations is treated as the ground truth of the spatial field, and the generated field intensities at the 824 sensor locations are treated as as the noise-free sensor reading from which noisy and distorted observations are generated. 

\subsection{Experimental Settings}

In this experiment, in addition to the five methods tested in the synthetic experiments, we evaluate the performance of the distributed version of S-BLUE, CEM, and ICM introduced in Section~\ref{sec:distributed}. To do so, we first divide the 824 sensor locations into 8 disjoint clusters via hierarchical clustering with great-circle distance and complete linkage (see e.g.\ Section~14.3.12 of \cite{friedman2001elements}). The 8 resulting clusters are indicated in Figure~\ref{fig:usmap} using 8 different colors. Subsequently, we apply the distributed approaches to this partitioned sensor network. 

In order to evaluate the reconstruction accuracy of the proposed methods, $25\%$ of the sensor locations are left out as the test set. For the 618 remaining locations, we randomly generate the distortion parameters from three different settings. 
In the three settings, each sensor has a respective probability of 0.3, 0.5 and 0.7 to introduce distortion.
Under all three settings, the distorting sensors have one of the three following categories, with equal probabilities. Under $\pi_1$, $a\sim\log\NORMAL(-0.4, 0.05^2)$, $b\sim\NORMAL(0, 0.2^2)$, under $\pi_2$, $a\sim\log\NORMAL(0.2, 0.05^2)$, $b\sim\NORMAL(0, 0.2^2)$, under $\pi_3$, $a\sim\log\NORMAL(0, 0.05^2)$, $b\sim\NORMAL(10, 2^2)$.
For each of the three settings, we first randomly generate the distortion parameters for each sensor. Then we randomly simulate noisy observations at each sensor location and subsequently apply the corresponding distortions. In addition, we examine the effect of the number of observations and effective SNR. Specifically, we examine the four following cases:
\begin{enumerate}
\item 10 observations per sensor, SNR=5dB;
\item 10 observations per sensor, SNR=15dB;
\item 50 observations per sensor, SNR=5dB;
\item 50 observations per sensor, SNR=15dB.
\end{enumerate}

\subsection{Results and Discussion}
Figure~\ref{fig:ustemp} shows the relative MSE averaged over 20 days, and Figure~\ref{fig:ustempfprfnr} shows the FPR and FNR of CEM and ICM, under each setting. 
Similar to the synthetic experiment~2, we observe that the MSE of the baselines and S-BLUE did not change with different number of observations per sensor. CEM and ICM, on the other hand, benefited from having access to more observations.

The naive baseline showed worse MSE when SNR was high. The reason was the same as in the synthetic experiment~2. In addition, the MSE of naive baseline is clearly affected by a higher proportion of distorting sensors.  

S-BLUE showed stable accuracies across all settings. Since S-BLUE does not estimate the distortion parameters, the error mainly resulted from smoothing. This can be seen clearly from the reconstructed spatial fields in Figure~\ref{fig:ustemprec}, which will be discussed later.

Compared to the simple method S-BLUE, the more sophisticated methods CEM and ICM have the additional benefit of being able to estimate the distortion parameters. 
CEM consistently outperformed ICM, which is consistent with what we observed in Section~\ref{ssec:realistic}. This possibly indicates the ineffectiveness of the iterative greedy search method with high-dimensional mixed discrete-continuous optimization problems. While CEM benefited slightly from higher SNR, the performance of ICM deteriorated with higher SNR when there were 10 observations per sensor, though they both benefited from more plentiful observations. Looking at Figure~\ref{fig:ustempfprfnr}, it can be observed that ICM had significantly higher FPR and slightly lower FNR compared to CEM. More observations and higher SNR have the effect of decreasing the FPR of both CEM and ICM. 

\begin{remark}
We noticed that if the SNR is set to be higher than 20dB, the performance of both CEM and ICM deteriorates greatly. With a high SNR (hence low noise variance), the posterior density of the model tends to be highly irregular. Thus the optimization problem might be highly ill-posed. In general, there is no universal solution to this. Some form of relaxation might help to improve the performance, but that is not examined in this work. 
\end{remark}

In the low SNR settings, the distributed approaches performed slightly worse compared to their centralized counterparts. In the high SNR settings, the distributed S-BLUE performed slightly worse compared to the centralized S-BLUE. However, the distributed CEM and ICM performed better compared to their centralized version. This indicates that both CEM and ICM suffered from convergence issues when dealing with the high-dimensional optimization problem. In the distributed version, the problem is decomposed into sub-problems of lower dimensionality, which are easier to solve. This experiment shows that the distributed approaches we proposed are the ideal candidates for solving such spatial field reconstruction problem in large-scale settings. 

Table~\ref{tab:realmaxdev} shows the maximum absolute deviation from the relative MSE in this experiment. It shows that S-BLUE and the distributed CEM had stable performance throughout the 20 days. ICM and the centralized CEM are less stable, presumably due to the difficulty in the optimization procedures. 

\begin{figure}[p]
\centering
\includegraphics[width=\linewidth]{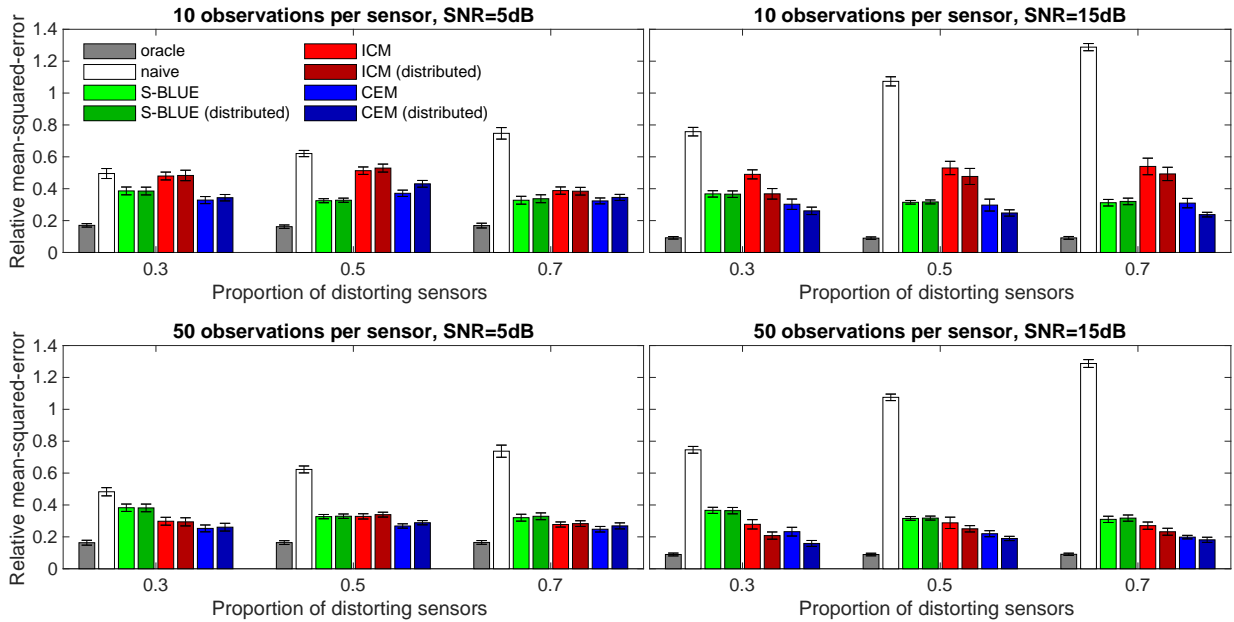}
\caption{Experiment with real data -- relative MSE against varying proportion of distorting sensors, number of observations and SNR.}
\label{fig:ustemp}
\end{figure}

\begin{figure}[p]
\centering
\includegraphics[width=\linewidth]{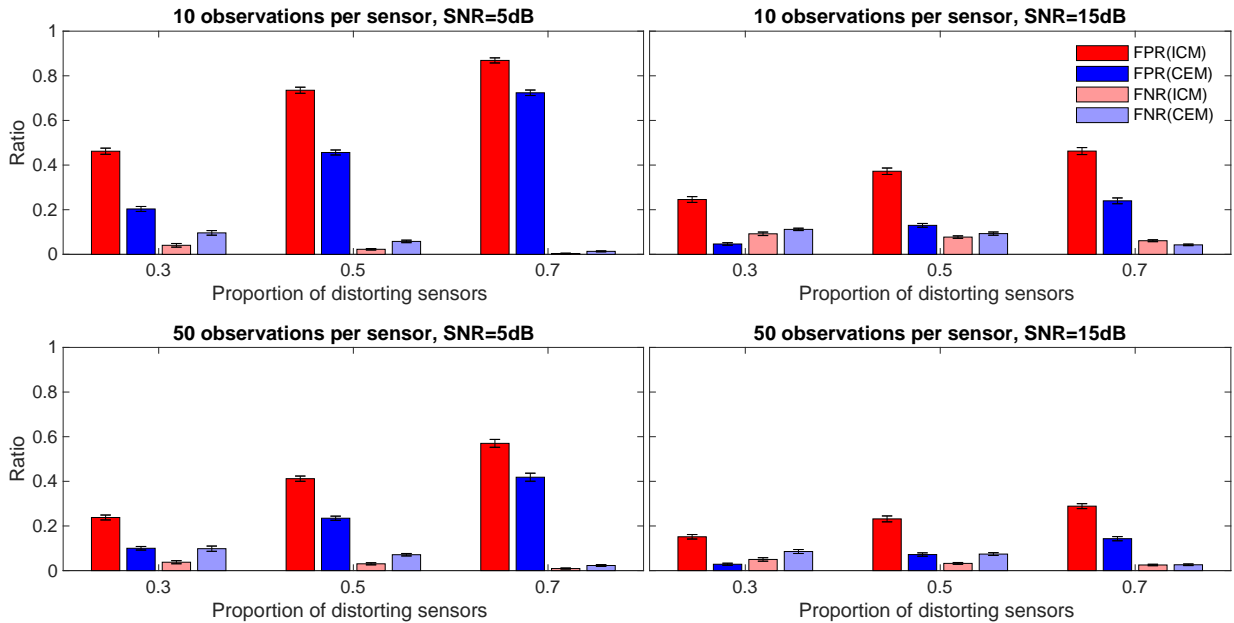}
\caption{Experiment with real data -- FPR, FNR of CEM and ICM.}
\label{fig:ustempfprfnr}
\end{figure}

\begin{table}[h]
\centering
\caption{Experiment with real data -- maximum absolute deviation from the average relative MSE of the eight methods. }
\label{tab:realmaxdev}
\begin{tabular}{|l|r|r|}
\hline 
Method & Centralized & Distributed \\ 
\hline 
oracle & 0.0728 & \multicolumn{1}{c|}{-} \\ 
\hline 
naive & 0.2015 & \multicolumn{1}{c|}{-} \\ 
\hline 
S-BLUE & 0.1132 & 0.1128 \\ 
\hline 
ICM & 0.2676 & 0.3188 \\ 
\hline 
CEM & 0.1838 & 0.1104 \\ 
\hline 
\end{tabular} 
\end{table}

Finally, let us examine the reconstructed spatial fields. Figure~\ref{fig:ustemprec} shows the heat maps of spatial fields reconstructed by centralized and distributed versions of S-BLUE, CEM, ICM, and the two baselines with the settings: proportion of distorting sensors is 0.7, 10 observations per sensor, SNR=15dB. The ground truth spatial field is also included for reference. First, notice that the ground truth contained much more details compared to the reconstructions, because all of the estimators here have the smoothing effect. Overall, the reconstruction of S-BLUE is much smoother and contains fewer details. This is due to nature of S-BLUE, as it does not estimate the distortion parameters. CEM and ICM, on the other hand, preserved many of the details, and their reconstructions were overall close to the one produced by the oracle. The naive method, however, produced a noticeably inaccurate reconstruction. 
As previously mentioned, the reconstructions of the distributed approaches are discontinuous and the discontinuity can be observed at the boundary of the clusters. Nonetheless, they produced accurate reconstructions of the spatial field.

Notice that the above analysis used only 20 consecutive days of data. When a longer time period is considered, the seasonal variation of the underlying spatial field must be taken into consideration. To demonstrate this, we use the data from each of the twelve months to estimate the GP hyperparameters by the same approach as above. Subsequently, the spatial field on the last day of each month is reconstructed by centralized CEM, using noisy and distorted sensor readings at the 618 locations generated in the same way as in the above experiment. These reconstructions along with the ground truths are shown in Figure~\ref{fig:ustempmonth}. It can be seen that there is an obvious seasonal effect on the temperature throughout the year. Notice that the reconstructions closely resemble the ground truths, which indicates that CEM worked well when the characteristics of the underlying spatial field varied throughout the year. 

\begin{figure}[t]
\centering
\includegraphics[width=\linewidth]{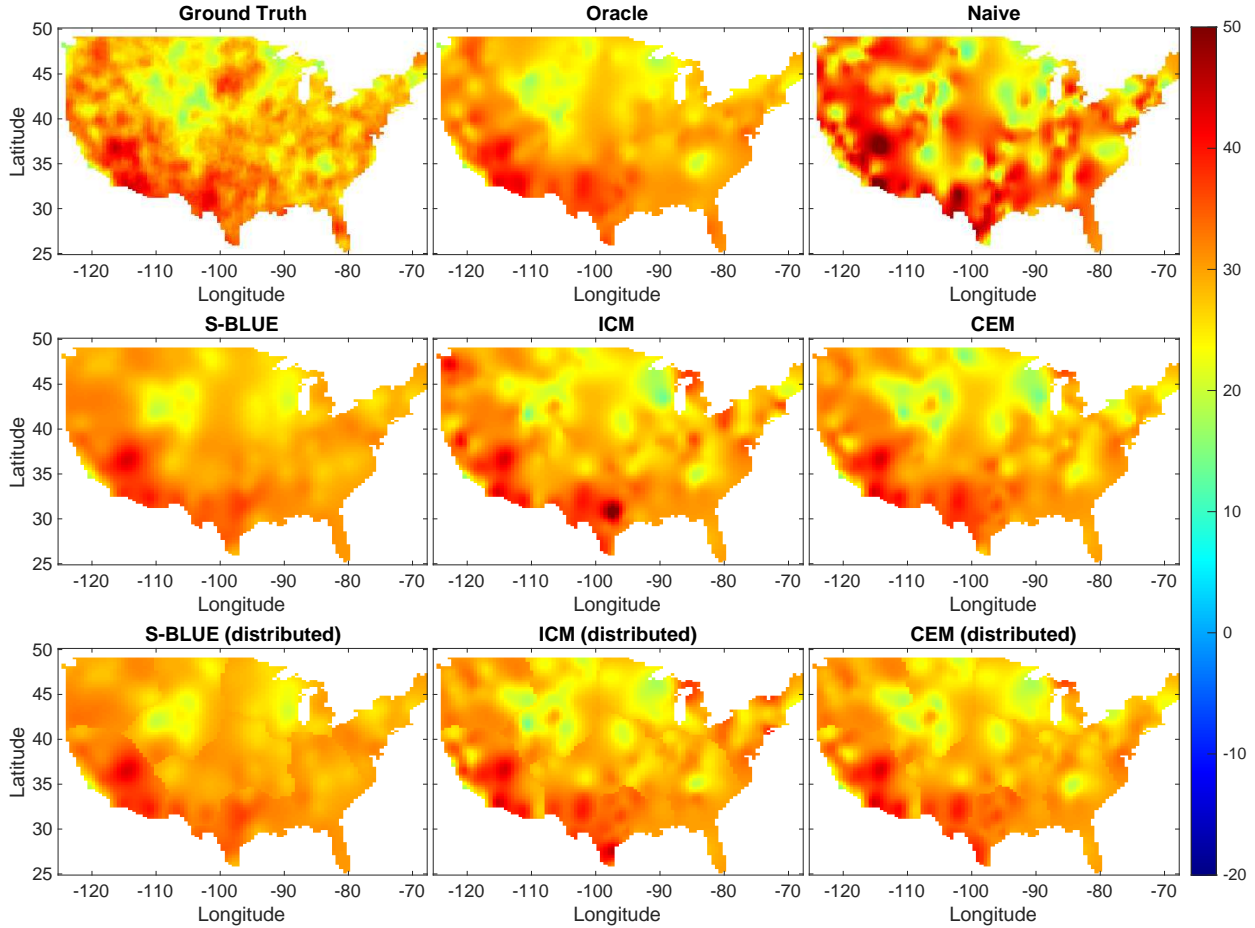}
\caption{Experiment with real data -- Reconstructed spatial fields on day 181. Settings: proportion of distorting sensors is 0.7, 10 observations per sensor, SNR=15dB.}
\label{fig:ustemprec}
\end{figure}

\begin{figure}[t]
\centering
\includegraphics[width=\linewidth]{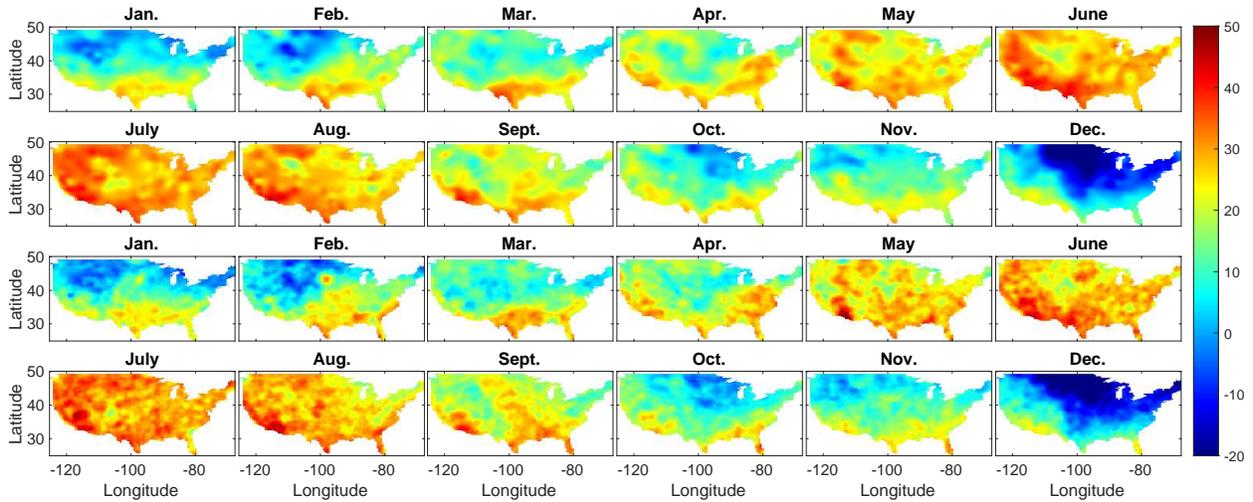}
\caption{Experiment with real data -- Reconstructed spatial fields on the last day of each month. The two top rows show the reconstructions by CEM, and the two bottom rows show the ground truths. Settings: proportion of distorting sensors is 0.7, 10 observations per sensor, SNR=15dB.}
\label{fig:ustempmonth}
\end{figure}

\section{Conclusion}
\label{sec:conclusion}
This paper addressed the problem of spatial field reconstruction based on distorted sensor readings. A new spatial field model based on a mixture of Gaussian process experts was developed. We developed two approaches to solve the inference problem. The first approach uses a linear Bayes estimator named the Spatial Best Linear Unbiased Estimator (S-BLUE), which is a low-complexity algorithm relying only on prior information. The second approach is a two-stage algorithm based on empirical Bayes, in which the unknown distortion parameters of the sensors are estimated based on distorted observations. We developed two optimization procedures for the two-stage algorithm, the first one is based on the Cross-Entropy method (CEM) and the second one is based on the Iterated Conditional Mode (ICM) which is an iterative greedy search procedure. 
In addition, the distributed versions of S-BLUE and empirical Bayes estimators were developed to improve the computational efficiency in large-scale applications.
We preformed two synthetic experiments as well as an experiment based on real temperature data from US EPA with synthetically generated distortions to assess the spatial field reconstruction accuracy of the proposed approaches. The results showed significant improvement compared to the estimation approach that neglects sensor distortions. 

\section*{Acknowledgments}
The research was conducted under the Cooling Singapore project, funded by Singapore’s National Research Foundation (NRF) under its Virtual Singapore programme. 

\bibliographystyle{IEEEtran}
\bibliography{references}

\appendices
\section{Proof of Theorems}

\subsection{Proof of Theorem~\ref{theorem:posterior}}
\label{apx:proofposterior}
\begin{proof}
For $n=1,\ldots,N$, the log model likelihood given $\BIf:=(f(\x_n))_{n=1:N}$ is given by
\begin{align*}
\begin{split}
&\log p(y_{n,1:M_n}|\BIf,\Bpsi)\\
=&-\frac{M_n}{2}\log2\pi-\frac{M_n}{2}\log a_n^2\varsigma^2-\frac{s_n}{2a_n^2\varsigma^2}\\
&+\frac{(a_nf_n+b_n)g_n}{a_n^2\varsigma^2}-\frac{M_n(a_nf_n+b_n)^2}{2a_n^2\varsigma^2},
\end{split}
\end{align*}
which depends on $y_{n,1:M_n}$ only through the statistics $g_n$ and $s_n$. Therefore, $(\BIg,\BIs)$ are sufficient for $(\BIf,\Bpsi)$. The joint density of $(\y,\BIf)$ conditional on $\Bpsi$ is given by
\begin{align}
\begin{split}
&\log p(\y,\BIf|\Bpsi)\\
=&\log p(\y|\BIf)+\log p(\BIf|\Bpsi)\\
=&-\frac{1}{2}\Big[\left(N+\tr(\BM)\right)\log2\pi+\tr(\BM\log(\varsigma^2\BA^2))+\log|\BCC|\\
&+\varsigma^{-2}\vecone^T\BA^{-2}\BIs+\Bmu^T\BCC^{-1}\Bmu+\varsigma^{-2}\BIb^T\BM\BA^{-2}\BIb-\Bgamma^T\BZ^{-1}\Bgamma\\
&-2\varsigma^{-2}\BIg^T\BA^{-2}\BIb+\left(\BIf-\BZ^{-1}\Bgamma\right)^T\BZ\left(\BIf-\BZ^{-1}\Bgamma\right)\Big],
\end{split}
\label{eqn:jointconditional}
\end{align}
where $\Bgamma=\varsigma^{-2}\BM\tilde{\BIg}+\BCC^{-1}\Bmu,\BZ=\varsigma^{-2}\BM+\BCC^{-1}$.
Integrating over $\BIf$, we deduce that,
\begin{align}
\begin{split}
\log p(\y|\Bpsi)=&\log\left[\int p(\y,\BIf|\Bpsi)\DIFF\BIf\right]\\
=&-\frac{1}{2}\left[\tr(\BM)\log2\pi+\tr(\BM\log(\varsigma^2\BA^2))\right.\\
&\left.+\log|\BCC|+\log|\BZ|+\varsigma^{-2}\vecone^T\BA^{-2}\BIs+\Bmu^T\BCC^{-1}\Bmu\right.\\
&\left.+\varsigma^{-2}\BIb^T\BM\BA^{-2}\BIb-2\varsigma^{-2}\BIg^T\BA^{-2}\BIb-\Bgamma^T\BZ^{-1}\Bgamma\right].
\end{split}
\label{eqn:logmlikstep1}
\end{align}
We have by Woodbury's formula that,
\begin{align}
\BZ^{-1}=&\varsigma^2\BM^{-1}-\varsigma^4\BM^{-1}(\BCC+\varsigma^2\BM^{-1})^{-1}\BM^{-1}\label{eqn:woodbury1}\\
=&\BCC-\BCC(\BCC+\varsigma^2\BM^{-1})^{-1}\BCC,\label{eqn:woodbury2}\\
\log|\BZ|=&\log|\varsigma^{-2}\BM|-\log|\BCC|+\log|\BCC+\varsigma^2\BM^{-1}|\label{eqn:woodbury3}.
\end{align}
We also have by $\Bgamma=\varsigma^{-2}\BM\tilde{\BIg}+\BCC^{-1}\Bmu$ that,
\begin{align}
\begin{split}
\Bgamma^T\BZ^{-1}\Bgamma=&\varsigma^{-4}\tilde{\BIg}^T\BM\BZ^{-1}\BM\tilde{\BIg}+2\varsigma^{-2}\tilde{\BIg}^T\BM\BZ^{-1}\BCC^{-1}\Bmu\\
&+\Bmu^T\BCC^{-1}\BZ^{-1}\BCC^{-1}\Bmu.
\end{split}
\label{eqn:logmlikstep1.1}
\end{align}
Substituting (\ref{eqn:woodbury1}) into the first term on the right-hand side of~(\ref{eqn:logmlikstep1.1}), using the fact that $\BM\BZ^{-1}\BCC^{-1}=(\BCC\BZ\BM^{-1})^{-1}=(\varsigma^{-2}\BCC+\BM^{-1})^{-1}=\varsigma^{2}\BUpsilon^{-1}$ for the second term of~(\ref{eqn:logmlikstep1.1}), substituting (\ref{eqn:woodbury2}) into the third term of~(\ref{eqn:logmlikstep1.1}), and finally substituting the resulting formula as well as (\ref{eqn:woodbury3}) into (\ref{eqn:logmlikstep1}), one gets the following after simplification (recall that $\tilde{\BIg}:=\BA^{-1}(\BM^{-1}\BIg-\BIb)$, $\BUpsilon:=\BCC+\varsigma^2\BM^{-1}$),
\begin{align*}
\begin{split}
\log p(\y|\Bpsi)=&-\frac{1}{2}\left[\tr(\BM)\log2\pi+\tr(\BM\log(\varsigma^2\BA^2))\right.\\
&-\log|\varsigma^2\BM^{-1}|+\log|\BUpsilon|+\varsigma^{-2}\vecone^T\BA^{-2}\BIs\\
&\left.-\varsigma^{-2}\BIg^T\BM^{-1}\BA^{-2}\BIg+(\tilde{\BIg}-\Bmu)^T\BUpsilon^{-1}(\tilde{\BIg}-\Bmu)\right].
\end{split}
\end{align*}
Notice that after integrating out $\BIf$, $(\BIg,\BIs)$ are still sufficient for $\Bpsi$. 

By Bayes' rule, $p(\Bpsi|\y)=\frac{p(\y|\Bpsi)\pi(\Bpsi)}{p(\y)}$,
where $p(\y)=\int p(\y|\Bpsi)\pi(\Bpsi)\DIFF\Bpsi$ is the normalizing constant that is analytically intractable. The proof is now complete. 
\end{proof}

\subsection{Proof of Theorem~\ref{theorem:postpred}}
\label{apx:proofpostpred}
From (\ref{eqn:jointconditional}) we see that $\log p(\y,\BIf|\Bpsi)$ has the following form,
\begin{align*}
\log p(\y,\BIf|\Bpsi)=q_1(\Bpsi)+q_2(\BIs,\Bpsi)+q_3(\BIg,\BIf,\Bpsi),
\end{align*}
where $q_1,q_2,q_3$ are some functions of the corresponding parameters. Similarly, we can deduce that $\log p(\y,\BIf,f_*|\Bpsi)$ has the following form,
\begin{align*}
\log p(\y,\BIf,f_*|\Bpsi)=q_1(\Bpsi)+q_2(\BIs,\Bpsi)+q_3(\BIg,[\BIf,f_*],\Bpsi).
\end{align*}
Therefore,
\begin{align}
\begin{split}
p(f_*|\y,\Bpsi)=&\frac{p(\y,f_*|\Bpsi)}{p(\y|\Bpsi)}=\frac{\int p(\y,\BIf,f_*|\Bpsi)\DIFF\BIf}{\int p(\y,\BIf|\Bpsi)\DIFF\BIf}\\
=&\frac{\int \exp(q_3(\BIg,[\BIf,f_*],\Bpsi))\DIFF\BIf}{\int \exp(q_3(\BIg,\BIf,\Bpsi))\DIFF\BIf}.
\end{split}
\end{align}
Thus, we have deduced that $p(f_*|\y,\Bpsi)$ depends on $\y$ only through $\BIg$ (or equivalently, $\tilde{\BIg}$), i.e.\ $p(f_*|\y,\Bpsi)=p(f_*|\BIg,\Bpsi)=p(f_*|\tilde{\BIg},\Bpsi)$.
We also have that conditional on $\Bpsi$, $(\tilde{\BIg},f_*)$ are jointly Gaussian, with
$\EXP[f_*|\Bpsi]=\mu_*$, 
$\EXP[\tilde{\BIg}|\Bpsi]=\Bmu$, 
$\COV[f_*|\Bpsi]=\CC_*$, 
$\COV[\tilde{\BIg},f_*|\Bpsi]=\BIk_*$, 
$\COV[\tilde{\BIg}|\Bpsi]=\BUpsilon$,
which can be easily verified. 
Thus, we have that $(f_*|\y,\Bpsi)\sim\NORMAL(\bar{f}_*,\sigma^2_*)$, where $\bar{f}_*=\mu_*+\BIk_*^T\BUpsilon^{-1}\left(\tilde{\BIg}-\Bmu\right)$, $\sigma^2_{*}=\CC_*-\BIk_*^T\BUpsilon^{-1}\BIk_*$.
The distribution of $(f_*,\Bpsi|\y)$ is given by $p(f_*,\Bpsi|\y)=p(f_*|\y,\Bpsi)p(\Bpsi|\y)$, 
and the posterior predictive distribution $(f_*|\y)$ is obtained by marginalizing over $\Bpsi$, $p(f_*|\y)=\int p(f_*|\y,\Bpsi)p(\Bpsi|\y)\DIFF\Bpsi$.
The proof is now complete.

\subsection{Proof of Theorem~\ref{theorem:sblue}}
\label{apx:proofsblue}
\begin{proof}
First, we claim that any linear estimator must have the form $h(\y)=\BIw^T\bar{\BIg}+b$, where $\BIw\in\R^N,b\in\R$. Due to the linearity of $h$ in observations $(y_{n,m})_{n=1:N,m=1:M_n}$, $\BIs$ is not involved. For $n=1,\ldots,N$, the weights of $(y_{n,m})_{m=1:M_n}$ must be the same due to symmetry. This proves the claim. 

Under quadratic loss, for $h\in\CH$, $R[\Pi,h]$ is given by
\begin{align*}
\begin{split}
R[\Pi,h]=&\EXP\left[\left(\BIw^T\bar{\BIg}+b-f_*\right)^2\right]\\
=&\BIw^T\EXP[\bar{\BIg}\bar{\BIg}^T]\BIw+b^2+\EXP[f_*^2]+2b\BIw^T\EXP[\bar{\BIg}]\\
&-2\BIw^T\EXP[f_*\bar{\BIg}]-2b\EXP[f_*].
\end{split}
\end{align*}
Differentiating $R[\Pi,h]$ with respect to $\BIw$ and $b$, we get
\begin{align*}
\frac{\partial R[\Pi,h]}{\partial\BIw}=&2\EXP\left[\bar{\BIg}\bar{\BIg}^T\right]\BIw+2b\EXP[\bar{\BIg}]-2\EXP[f_*\bar{\BIg}],\\
\frac{\partial R[\Pi,h]}{\partial b}=&2b+2\BIw^T\EXP[\bar{\BIg}]-2\EXP[f_*].
\end{align*}
Setting the partial derivatives to $\veczero$ and solving the equations gives the optimal weight vector and intercept,
\begin{align}
\widehat{b}=&\EXP[f_*]-\widehat{\BIw}^T\EXP[\bar{\BIg}],\\
\begin{split}
\widehat{\BIw}=&\left(\EXP[\bar{\BIg}\bar{\BIg}^T]-\EXP[\bar{\BIg}]\EXP[\bar{\BIg}]^T\right)^{-1}\left(\EXP[f_*\bar{\BIg}]-\EXP[f_*]\EXP[\bar{\BIg}]\right)\\
=&\COV[\bar{\BIg}]^{-1}\COV[\bar{\BIg},f_*].
\end{split}
\end{align}
One can verify that $(\widehat{\BIw},\widehat{b})$ is indeed minimizing the Bayes risk. Hence,
\begin{align}
\begin{split}
\widehat{h}_{\text{S-BLUE}}(\y)=&\widehat{\BIw}^T\bar{\BIg}+\widehat{b}\\
=&\EXP[f_*]+\COV[\bar{\BIg},f_*]^T\COV[\bar{\BIg}]^{-1}(\bar{\BIg}-\EXP[\bar{\BIg}]).
\end{split}
\end{align}

The terms $\EXP[f_*],\EXP[\bar{\BIg}],\COV[\bar{\BIg},f_*],\COV[\bar{\BIg}]$ can all be expressed in closed-form. The closed-form expressions and the details of the computation are given below. Let $\odot$ denote matrix entry-wise product. 
\begin{align}
\EXP[f_*]=&\mu_*,\\
\begin{split}
\EXP[\bar{\BIg}]=&\EXP[\EXP[\bar{\BIg}|\Bpsi]]=\diag\left(\EXP[\BIa]\right)\Bmu+\EXP[\BIb],
\end{split}\\
\begin{split}
\COV[\bar{\BIg},f_*]=&\EXP[\COV[\bar{\BIg},f_*|\Bpsi]]+\COV[\EXP[\bar{\BIg}|\Bpsi],\mu_*]\\
=&\diag\left(\EXP[\BIa]\right)\BIk_*,
\end{split}\\
\begin{split}
\COV[\bar{\BIg}]=&\EXP[\COV[\bar{\BIg}|\Bpsi]]+\COV[\EXP[\bar{\BIg}|\Bpsi]]\\
=&\EXP[\BA\BCC\BA+\varsigma^2\BM^{-1}\BA^{2}]+\COV[\BA\Bmu+\BIb]\\
=&\EXP[\BIa\BIa^T]\odot\left(\BCC+\varsigma^2\BM^{-1}+\Bmu\Bmu^T\right)\\
&+\diag(\Bmu)\left(\EXP[\BIa\BIb^T]+\EXP[\BIa\BIb^T]^T\right)\\
&+\EXP[\BIb\BIb^T]-\EXP[\bar{\BIg}]\EXP[\bar{\BIg}]^T,
\end{split}
\end{align}

In the terms $\EXP[\BIa]$, $\EXP[\BIb]$, $\EXP[\BIa\BIa^T]$, $\EXP[\BIb\BIb^T]$, $\EXP[\BIa\BIb^T]$, the expectations are evaluated entry-wise.
For example, entries of $\EXP[\BIa]$ are given by
\begin{align}
\begin{split}
\EXP\left[a_n\right]=&q^{(n)}_0+\sum_{k=1}^Kq^{(n)}_k\EXP\left[a_n|Z_n=k\right].
\end{split}
\label{eqn:sbluea}
\end{align}
Entries of $\EXP[\BIa\BIa^T]$ are given by
\begin{align}
\begin{split}
\EXP\left[a_ia_j\right]=&\begin{cases}
\sum\limits_{k=0}^Kq^{(i)}_k\EXP\left[{a_i}^2|Z_i=k\right]&\text{if}\;i=j,\\
\sum\limits_{k=0}^K\sum\limits_{k'=0}^Kq^{(i)}_kq^{(j)}_{k'}\EXP\left[a_i|Z_i=k\right]&\text{if}\;i\neq j. \\
\quad\quad\times \EXP\left[a_j|Z_j=k'\right]
\end{cases}
\end{split}
\label{eqn:sblueaa}
\end{align}
Entries of $\EXP[\BIb]$, $\EXP[\BIb\BIb^T]$ and $\EXP[\BIa\BIb^T]$ can be evaluated similarly. 
With the above equations, we are able to evaluate $\widehat{h}_{\text{S-BLUE}}$ efficiently.
\end{proof}

\end{document}